\documentclass[letterpaper]{article} 
\usepackage{aaai23}  
\usepackage{times}  
\usepackage{helvet}  
\usepackage{courier}  
\usepackage[hyphens]{url}  
\usepackage{graphicx} 
\usepackage{natbib}  
\usepackage{caption} 
\usepackage{algorithm}
\usepackage[noend]{algorithmic}

\usepackage{newfloat}
\usepackage{listings}
\DeclareCaptionStyle{ruled}{labelfont=normalfont,labelsep=colon,strut=off} 
\floatstyle{ruled}
\newfloat{listing}{tb}{lst}{}
\floatname{listing}{Listing}

\usepackage{amsmath}
\usepackage{amssymb}
\usepackage{subcaption}
\usepackage{amsthm}
\usepackage{accents}
\usepackage{tikz}
\usetikzlibrary{calc}
\usetikzlibrary{shapes}

\usepackage{caption}
\usepackage{subcaption}

\newtheorem{theorem}{Theorem}
\newtheorem{lemma}{Lemma}
\newtheorem{definition}{Definition}

\newcommand{\game}{\mathsf{G}}

\newcommand{\UBV}{\overline{V}}
\newcommand{\LBV}{\underline{V}}
\newcommand{\AUBV}{\tilde{V}}
\newcommand{\ALBV}{\undertilde{V}}
\newcommand{\dom}{Dom}

\newcommand{\optval}{\textsf{OPT}}

\DeclareMathOperator*{\argmax}{arg\,max}
\DeclareMathOperator*{\argmin}{arg\,min}

\title{Function Approximation for Solving Stackelberg Equilibrium in Large Perfect Information Games}
\author {
    Chun Kai Ling,
    J. Zico Kolter,
    Fei Fang
}
\affiliations {
    School of Computer Science, Carnegie Mellon University \\
    chunkail@cs.cmu.edu, zkolter@cs.cmu.edu, feif@cs.cmu.edu
}

\usepackage{bibentry}

\begin{document}

\maketitle

\begin{abstract}
Function approximation (FA) has been a critical component in solving large zero-sum games.
    Yet, little attention has been given towards FA in solving \textit{general-sum} extensive-form games, despite them being widely regarded as being computationally more challenging than their fully competitive  or cooperative counterparts. 
    A key challenge is that for many equilibria in general-sum games, no simple analogue to the state value function used in Markov Decision Processes and zero-sum games exists.
    In this paper, we propose learning the \textit{Enforceable Payoff Frontier} (EPF)---a generalization of the state value function for general-sum games.
    We approximate the optimal \textit{Stackelberg extensive-form correlated equilibrium} by representing EPFs with neural networks and training them by using appropriate backup operations and loss functions. 
    This is the first method that applies FA to the Stackelberg setting, allowing us to scale to much larger games while still enjoying performance guarantees based on FA error. Additionally, our proposed method guarantees incentive compatibility and is easy to evaluate without having to depend on self-play or approximate best-response oracles.
\end{abstract}

\frenchspacing

\section{Introduction}
\label{sec:intro}
A central challenge in modern game solving is to handle large game trees, particularly those too large to traverse or even specify. These include board games like Chess, Poker \cite{silver2018general,silver2016mastering,brown2017libratus,brown2019superhuman,moravvcik2017deepstack,bakhtin2021no,gray2020human} and modern video games with large state and action spaces \cite{alphastarblog}.
Today, scalable game solving is frequently achieved via function approximation (FA), typically by using neural networks to model state values and harnessing the network's ability to generalize its evaluation to states never encountered before \cite{silver2016mastering,silver2018general,moravvcik2017deepstack,schmid2021player}. Methods employing FA have achieved not only state-of-the-art performance, but also exhibit more human-like behavior \cite{kasparov2018chess}.

Surprisingly, FA is rarely applied to solution concepts used in general-sum games such as Stackelberg equilibrium, which are generally regarded as being more difficult to solve than the perfectly cooperative/competitive Nash equilibrium. Indeed, the bulk of existing literature centers around on methods such as exact backward induction \cite{DBLP:journals/corr/BosanskyBHMS15, bovsansky2017computation}, incremental strategy generation \cite{vcerny2018incremental,cermak2016using,karwowski2020double}, and mathematical programming \cite{bosansky2015sequence}.\footnote{Meta-game solving \cite{lanctot2017unified,wang2019deep} is used in zero-sum games, but not general-sum Stackelberg games.} While exact, these methods rarely scale to large game trees, especially those too large to traverse, severely limiting our ability to tackle general-sum games that are of practical interest, such as those in security domains like wildlife poaching prevention \cite{fang2017paws} and airport patrols \cite{pita2008armor}. 
For non-Nash equilibrium in general-sum games, the value of a state often \textit{cannot} be summarized as a scalar (or fixed sized vector), rendering the direct application of FA-based zero-sum solvers like \cite{silver2018general} infeasible. 

In this paper, we propose applying FA to model the \textit{Enforceable Payoff Frontier} (EPF) for each state and using it to solve for the \textit{Stackelberg extensive-form correlated equilibrium} (SEFCE) in two-player games of perfect information. Introduced in \cite{bovsansky2017computation,DBLP:journals/corr/BosanskyBHMS15,letchford2010computing}, EPFs capture the tradeoff between player payoffs and is analogous to the state value in zero-sum games.\footnote{The idea of an EPF was initially used by \cite{letchford2010computing} to give a polynomial time solution for SSEs. However, they (as well as other work) do not propose any naming.} 
Specifically, we (i) study the pitfalls that can occur with using FA in general-sum games, (ii) propose a method for solving SEFCEs by modeling EPFs using neural networks and minimizing an appropriately designed Bellman-like loss, and (iii) provide guarantees on incentive compatibility and performance of our method. Our approach is the first application of FA in Stackelberg settings without relying on best-response oracles for performance guarantees. Experimental results show that our method can (a) approximate solutions in games too large to explicitly traverse, and (b) generalize learned EPFs over states in a repeated setting where game payoffs vary based on features. 

\section{Preliminaries and Notation}
\label{sec:background}
A 2-player \textit{perfect information} game $\game$ 
is represented by a finite game tree with game states $s \in \mathcal{S} $ given by vertices and action space $\mathcal{A}(s)$ given by directed edges starting from $s$. Each state belongs to either player $\mathsf{P}_1$ or $\mathsf{P}_2$; we denote these disjoint sets by $\mathcal{S}_1$ and $\mathcal{S}_2$ respectively.
Every leaf (terminal state) $\ell \in \mathcal{L} \subseteq \mathcal{S}$ of $\game$ is associated with payoffs, given by $r_i(\ell)$ for each player $i$. 
Taking action $a \in \mathcal{A}(s)$ at state $s \not \in \mathcal{L}$ leads to $s'=T(a; s)$, where $s' \in \mathcal{S}$ is the next state and $T$ is the deterministic transition function. Let $\mathcal{C}(s) = \{ s' \mid T(a; s) = s', a \in \mathcal{A}(s)\}$ denote the immediate children of $s$.
We say that state $s$ precedes ($\sqsubset$) state $s'$ if $s \neq s'$ and $s$ is an ancestor of $s'$ in $\game$, and write $\sqsubseteq$ if allowing $s=s'$. 
An action $a \in \mathcal{A}(s)$ \textit{leads to} $s'$ if $s \sqsubset s'$ and $T(a; s) \sqsubseteq s'$. 
With a slight abuse of notation, we denote $T(a; s) \sqsubseteq s'$ by $a \sqsubset s'$ or $(s, a) \sqsubset s'$. Since $\game$ is a tree, for states $s, s'$ where $s \sqsubset s'$, exactly one $a \in \mathcal{A}(s)$ such that $(s, a) \sqsubseteq s'$. We use the notation $\sqsupseteq$ and $\sqsupset$ when the relationships are reversed.

A strategy $\pi_i$, where $i \in \{ \mathsf{P}_1, \mathsf{P}_2 \}$, is a mapping from state $s \in \mathcal{S}_i$ to a distribution over actions $\mathcal{A}(s)$, i.e., $\sum_{a \in \mathcal{A}(s)} \pi_i(a; s) = 1$. 
Given strategies $\pi_1$ and $\pi_2$, the probability of reaching $\ell \in \mathcal{L}$ starting from $s$ is 
given by $p(\ell| s ; \pi_1, \pi_2) = \prod_{i \in \{ \mathsf{P}_1, \mathsf{P}_2 \} } \prod_{(s',a); s \sqsubseteq s', (s',a) \sqsubset \ell, s' \in \mathcal{S}_i} \pi_i(a; s')$, and player $i$'s expected payoff starting from $s$ is $R_i(s ; \pi_1, \pi_2) = \sum_{\ell \in \mathcal{L}} p(\ell | s ; \pi_1, \pi_2) r_i(\ell)$. 
We use as shorthand $p(\ell; \pi_1, \pi_2)$ and $R_i(\pi_1, \pi_2)$ if $s$ is the root. 
A strategy $\pi_{2}$ is a \textit{best response} to a strategy $\pi_{1}$ if 
$R_{2}(\pi_1, \pi_{2}) \geq R_{2}(\pi_1, \pi'_{2})$ for all strategies $\pi'_{2}$. The set of best responses to $\pi_1$ is written as $BRS_2(\pi_1)$. 

The \textit{grim strategy} $\argmin_{\pi_1} \max_{\pi_2} R_2(\pi_1, \pi_2)$ of $\mathsf{P}_1$ towards $\mathsf{P}_2$ is one which guarantees the lowest payoff for $\mathsf{P}_2$. Conversely, the \textit{joint altruistic strategy} $\argmax_{\pi_1, \pi_2} R_2(\pi_1, \pi_2)$ is one which maximizes $\mathsf{P}_2$'s payoff.
We restrict grim and altruistic strategies to those which are subgame-perfect, i.e., they remain the optimal if the game was rooted at some other state.\footnote{This is to avoid strategies which play arbitrarily at states which have $0$ probability of being reached.}
Grim and altruistic strategies ignore $\mathsf{P}_1$'s own payoffs and can be computed by backward induction. For each state, we denote by $\underline{V}(s)$ and $\overline{V}(s)$ the internal values of $\mathsf{P}_2$ for grim and altruistic strategies obtained via backward induction.

\subsection{Stackelberg Equilibrium in Perfect Information Games}
In a Strong Stackelberg equilibrium (SSE), there is a distinguished \textit{leader} and \textit{follower}, which we assume are $\mathsf{P}_1$ and $\mathsf{P}_2$ respectively. The leader \textit{commits} to any strategy $\pi_1$ and the follower best responds to the leader, breaking ties by selecting $\pi_2 \in BRS_2(\pi_1)$ such as to benefit the leader. 
\footnote{Commitment rights are justified by repeated interactions. If the $\mathsf{P}_1$ reneges on its commitment, $\mathsf{P}_2$ plays another best response, which is detrimental to the leader. This setting is unlike \cite{de2020strategic} which uses binding agreements.}
Solving for the SSE entails finding the optimal commitment for the \textit{leader}, i.e., a pair $\pi =(\pi_1, \pi_2)$ such that $\pi_2 \in BRS_2(\pi_1)$ and $R_1(\pi_1, \pi_2)$ is to be maximized. 

\begin{figure*}[t]
    \centering
    \begin{subfigure}[b]{.24\textwidth}
    \centering
    \begin{tikzpicture}[scale=1,font=\footnotesize]
\tikzstyle{leader}=[regular polygon,regular polygon sides=50,draw,inner sep=1.2,minimum size=12];
\tikzstyle{follower}=[regular polygon,regular polygon sides=50, rotate=180,draw,inner sep=1.2, minimum size=12, fill=gray!65];
\tikzstyle{terminal}=[draw,inner sep=1.2];
\tikzstyle{level 1}=[level distance=9mm,sibling distance=20mm]
\tikzstyle{level 2}=[level distance=5mm,sibling distance=12mm]

\node(0)[follower, label=center:{$s$}]{}
    child{node(1)[leader, label=center:{$s'$}]{}
        child{node(3)[terminal, label=below:{$(10,-1)$}]{}
        }
        child{node(4)[terminal, label=below:{$(-1,1)$}]{}
        }
        edge from parent node[above left]{stay}
    }
    child{node(2)[terminal, label=below:{$(k_1,k_2)$}]{}
    edge from parent node[above right]{exit}
    };
\end{tikzpicture}
    \caption{Toy example.}
    \label{fig:basic_eg}
    \end{subfigure}
    \begin{subfigure}[b]{.24\linewidth}
    \centering
    \includegraphics[width=\textwidth]{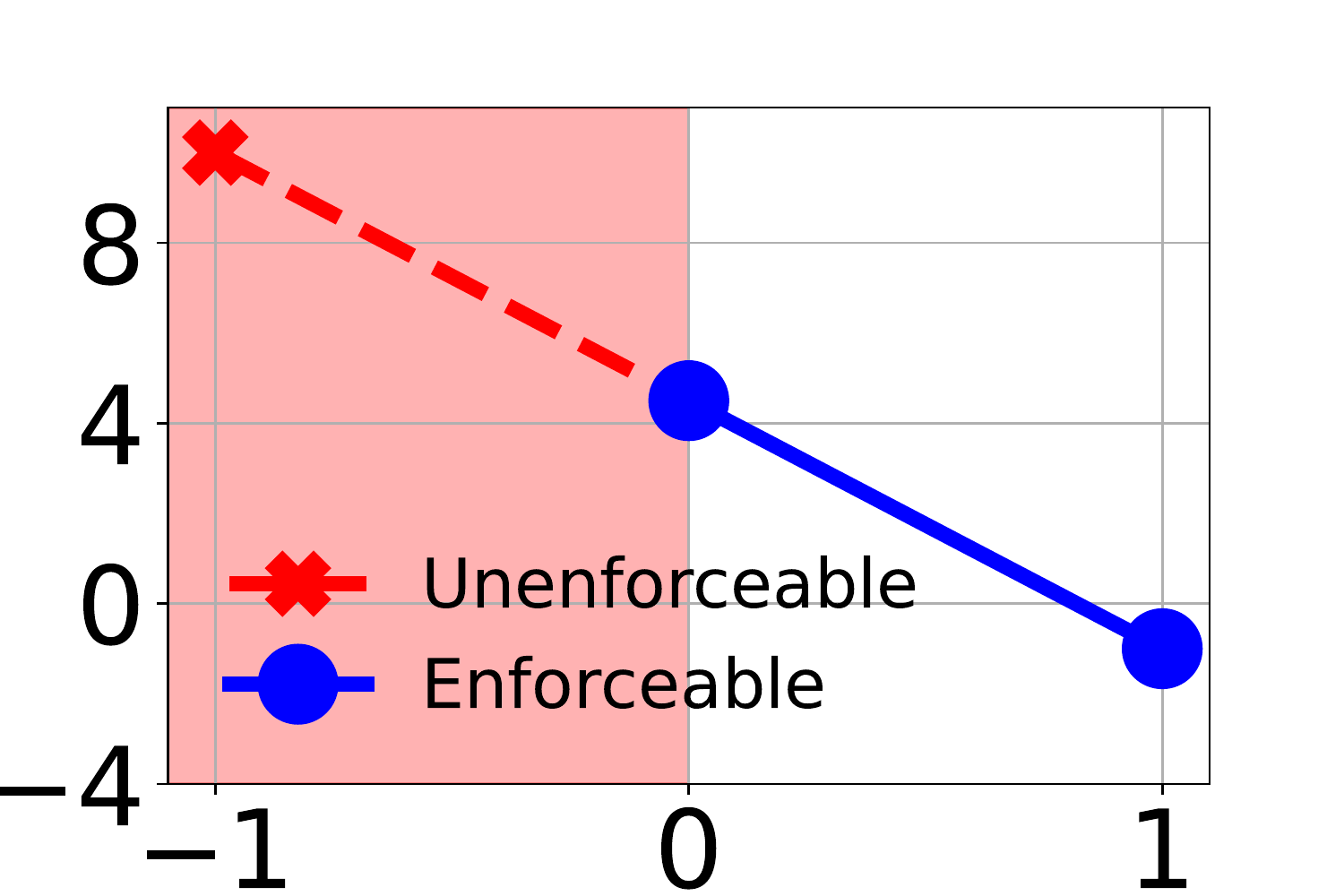}
    \subcaption{EPF at vertex $s'$.}
    \label{fig:demo-2}
    \end{subfigure}
    \begin{subfigure}[b]{.24\linewidth}
    \centering
    \includegraphics[width=\textwidth]{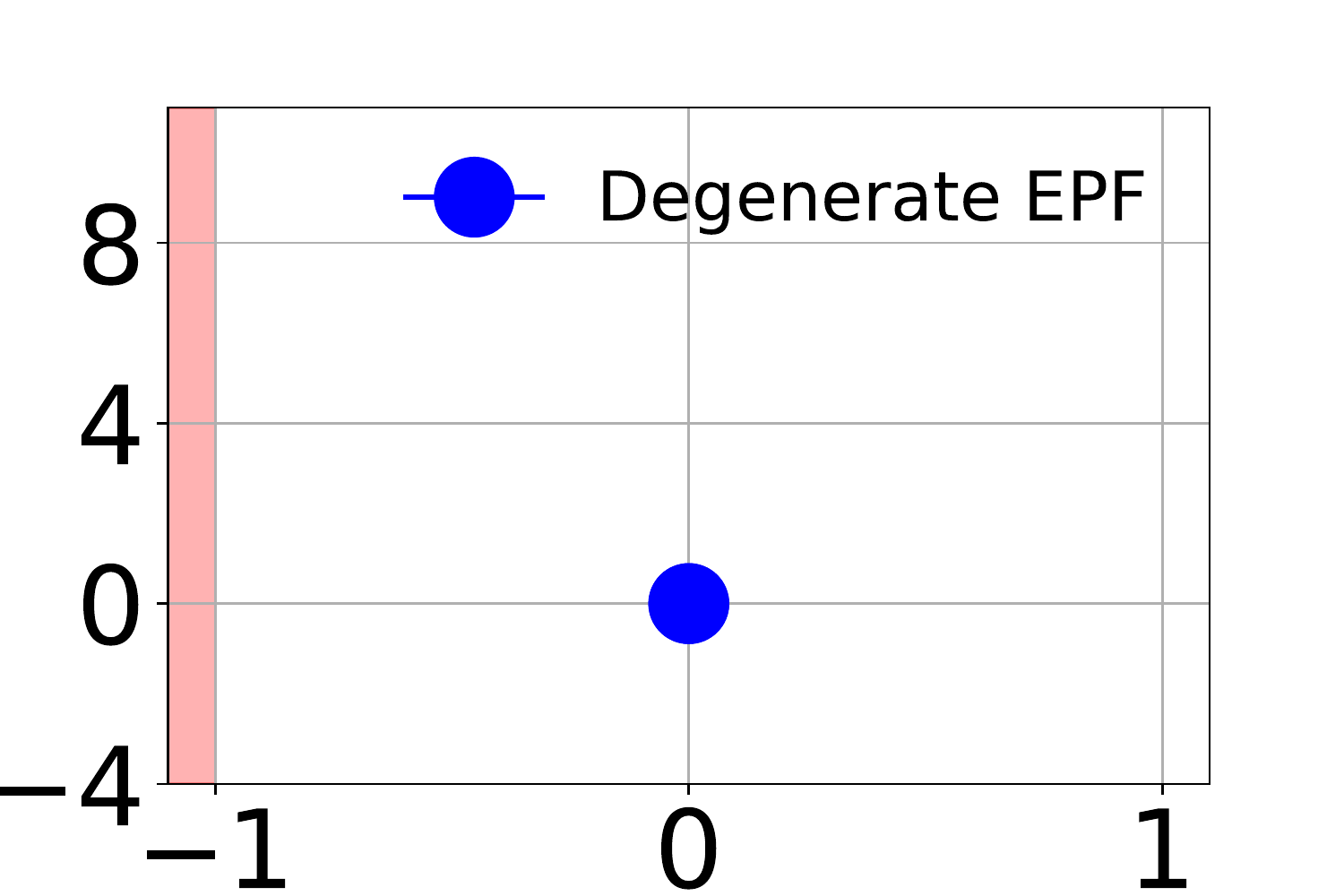}
    \subcaption{EPF at vertex after exiting.}
    \label{fig:demo-3}
    \end{subfigure}
    \begin{subfigure}[b]{.24\linewidth}
    \centering
    \includegraphics[width=\textwidth]{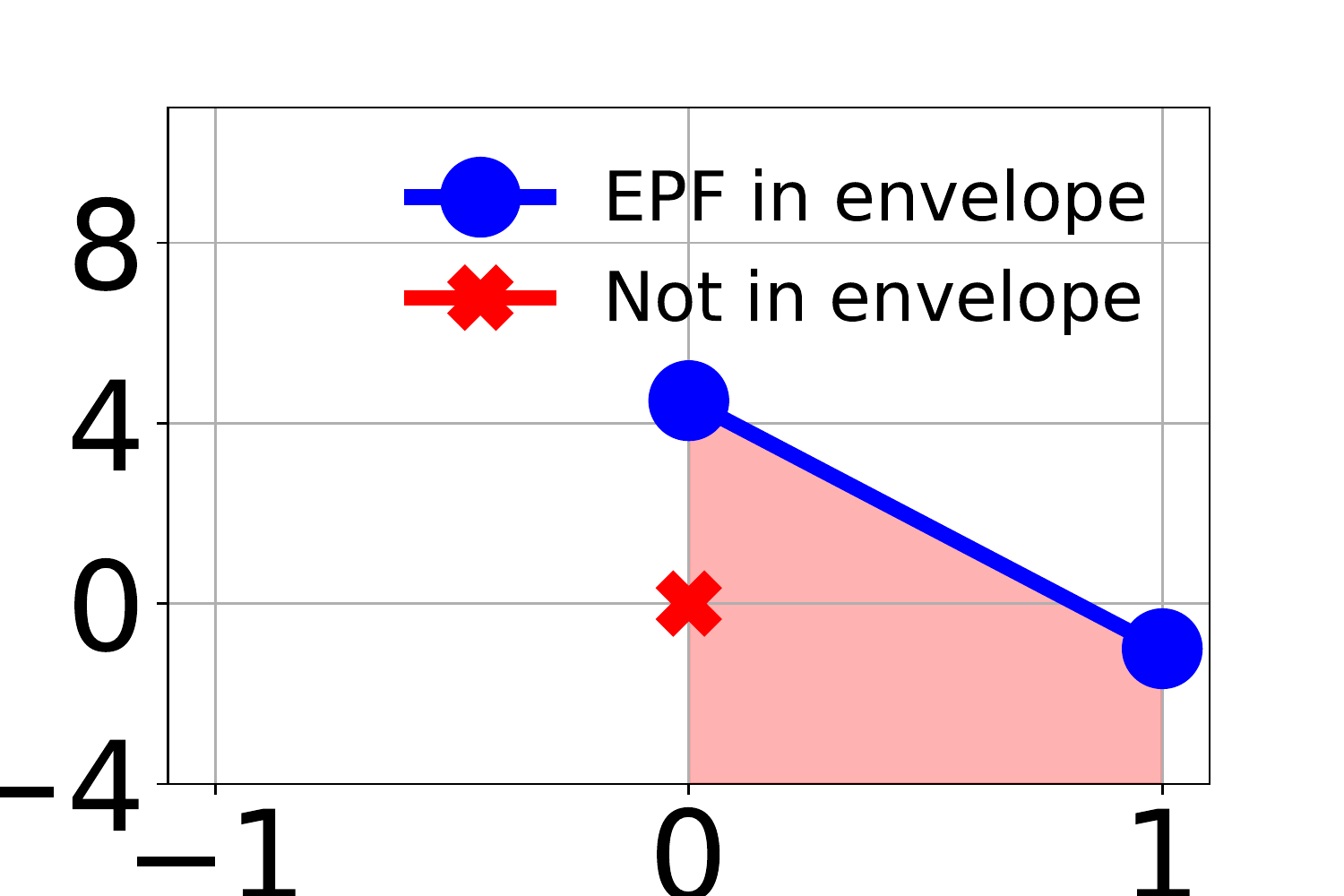}
    \subcaption{EPF at root vertex $s$.}
    \label{fig:demo-4}
    \end{subfigure}
    \caption{(a) Game tree to illustrate computation of SEFCE. Leader \tikz\draw[black] (0,0) circle (.7ex);, follower  \tikz\draw[black,fill=gray] (0,0) circle (.7ex); and leaf $\square$ states are vertices and edges are actions. 
    (b-d) EPFs at $s'$, after exiting and $s$. The x and y axes are follower ($\mu_2$) and leader payoffs ($U_s(\mu)$). In (b) the pink regions give $\mathsf{P}_2$ too little and are truncated. In (d), the pink regions are not part of the upper concave envelope and removed.}
\end{figure*}
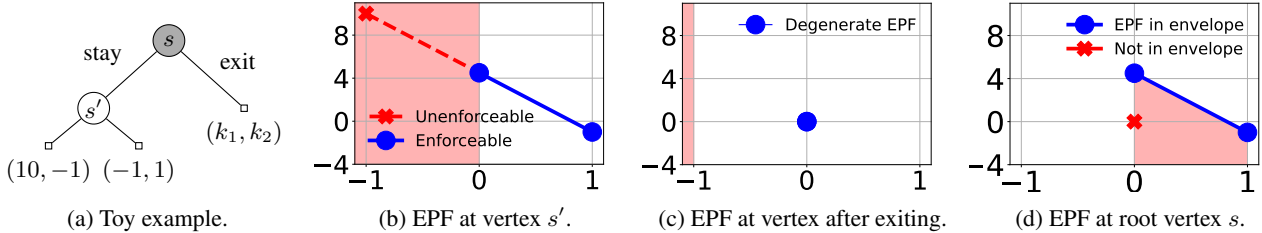

It is well-known that the optimal SSE will perform no worse (for the leader) than Nash equilibrium, and often much better. Consider the game in Figure~\ref{fig:basic_eg} with $k_1=k_2=0$. If the expected follower payoff from staying is less than $0$, then it would exit immediately. Hence, solutions such as the subgame perfect Nash gives a leader payoff of $0$. The optimal Stackelberg solution is for the leader to commit to a uniform strategy---this ensures that staying yields the follower a payoff of $0$, which under the tie-breaking assumptions of SSE nets the leader a payoff of $4.5$. 

\paragraph{Stackelberg Extensive-Form Correlated Equilibrium} 
For this paper, we will focus on a relaxation of the SSE known as the Stackelberg extensive-form correlated equilibirum (SEFCE), which allows the leader to explicitly recommend actions to the follower \textit{at the time of decision making}. If the follower deviates from the recommendation, the leader is free to retaliate---typically with the grim strategy. 
In a SEFCE, $\mathsf{P}_1$ takes and recommends actions to maximize its reward, subject to the constraints that the recommendations are sufficiently appealing to $\mathsf{P}_2$ relative to threat of $\mathsf{P}_2$ facing the grim strategy after any potential deviation.
\begin{definition}[Minimum required incentives]
Given $s \in \mathcal{S}_2$, $s'\in\mathcal{C}(s)$,
we define the \textit{minimum required incentive} $\tau(s')=\max_{s^! \in \mathcal{C}(s); s^! \neq s'} \LBV(s^!)$, i.e., the minimum amount that $\mathsf{P}_1$ needs to \textit{promise} $\mathsf{P}_2$ under $s'$ for it to be reached. 
\label{def:approx_min_req_incentive}
\end{definition}
\begin{definition}[SEFCE]
A strategy pair $\pi=(\pi_1, \pi_2)$ is a SEFCE if it is incentive compatible, i.e., for all $s \in \mathcal{S}_2$, $a \in \mathcal{A}(s)$, $\pi_2(a;s) > 0 \Longrightarrow R_2(T(a; s); \pi_1, \pi_2) \geq \tau(T(a; s))$. Additionally, $\pi$ is optimal if $R_1(\pi_1, \pi_2)$ is maximized.
\label{def:sefce}
\end{definition}
In Section~\ref{sec:review}, we describe how optimal SEFCE can be computed in polynomial time for perfect information games. 


\subsection{Function Approximation of State Values}
When finding Nash equilibrium in perfect information games, the \textit{value} $v_s$  of a state is a crucial quantity which summarizes the utility obtained from $s$ onward, assuming optimal play from all players. It contains sufficient information for one to obtain an optimal solution after using them to `replace' subtrees. Typically $v_s$ should only rely on states $s' \sqsupseteq s$. In zero-sum games, $v_s=\underline{V}_s$ while in cooperative games, $v_s=\overline{V}_s$. Knowing the true value of each state immediately enables the optimal policy via one-step lookahead. While $v_s$ can be computed over all states by backward induction, this is not feasible when $\game$ is large. A standard workaround is to replace $v_s$ with an approximate $\tilde{v}_s$ which is then used in tandem with some search algorithm (depth-limited search, Monte-Carlo tree search, etc.) to obtain an approximate solution. 
Today, $\tilde{v}_s$ is often \textit{learned}. By representing $\tilde{v}$ with a rich function class over state features (typically using a neural network), modern solvers are able to generalize $\tilde{v}$ across large state spaces without explicitly visiting every state, thus scaling to much larger games.

\paragraph{Fitted Value Iteration.} A class of methods closely related to ours is Fitted Value Iteration (FVI) \cite{lagoudakis2003least,dietterich2001batch,munos2008finite}. The idea behind FVI is to optimize for parameters such as to minimize the \textit{Bellman loss} over sampled states by treating it as a regular regression problem. \footnote{We distinguish RL and FVI in that the transition function is known explicitly and made used of in FVI.} Here, the Bellman loss measures the distance between $\tilde{v}_s$ and the estimated value using one-step lookahead using $\tilde{v}$. If this distance is $0$ for all $s$, then $\tilde{v}$ matches the optimal $v$. In practice, small errors in FA accumulate and cascade across states, lowering performance. Thus, it is important to bound performance as a function of the Bellman loss over all $s$.

\subsection{Related Work}
Some work has been done in generalizing state values in general-sum games, but few involve learning them. Related to ours is \cite{murray2007finding,macdermed2011quick,dermed2013value}, which approximate the achievable set of payoffs for correlated equilibrium, and eventually SSE \cite{letchford2012computing} in stochastic games. These methods are analytical in nature and scale poorly. \cite{perolat2017learning,greenwald2003correlated} propose a Q-learning-like algorithm over general-sum Markov games, but do not apply FA and only consider stationary strategies which preclude strategies involving long range threats like the SSE. \cite{zinkevich2005cyclic} show a class of general-sum Markov games where value-iteration like methods will necessarily fail. \cite{zhong2021can} study reinforcement learning in the Stackelberg setting, but only consider followers with myopic best responses. \cite{castelletti2011multi} apply FVI in a multiobjective setting, but do not consider the issue of incentive compatibility. 
Another approach is to apply reinforcement learning and self-play \cite{leibo2017multi}. Recent methods account for the nonstationary environment each player faces during training \cite{foerster2017learning,perolat2022mastering}; however they have little game theoretical guarantees in terms of incentive compatibility, particularly in non zero-sum games. 


\section{Review: Solving SEFCE via Enforceable Payoff Frontiers}
\label{sec:review}
In Section~\ref{sec:background}, we emphasized the importance of the value function $v$ in solving zero-sum games.
In this section, we review the analogue for SEFCE in the general-sum games, which we term as \textit{Enforceable Payoff Frontiers} (EPF). Introduced in \cite{letchford2010computing}, the EPF at state $s$ is a \textit{function} $U_s: \mathbb{R}\mapsto \mathbb{R}\cup\{ -\infty\}$, such that $U_s(\mu_2)$ gives the maximum leader payoff for a SEFCE for a game rooted at $s$, on condition that $\mathsf{P}_2$ obtains a payoff of $\mu_2$.
All leaves $s \in \mathcal{L}$ have degenerate EPFs $U_s(r_2(s))=r_1(s)$ and $-\infty$ everywhere else. EPFs capture the tradeoff in payoffs between $\mathsf{P}_1$ and $\mathsf{P}_2$, making them useful for solving SEFCEs. We now review the two-phase algorithm of \cite{bovsansky2017computation} using the example game in Figure~\ref{fig:basic_eg} with $k_1=k_2=0$. This approach forms the basis for our proposed FA method.

\paragraph{Phase 1: Computing EPF by Backward Induction.} The EPF at $s'$ is given by the line segment connecting payoffs of its children EPF and $-\infty$ everywhere else. This is because the leader is able to freely mix over actions.
To compute $U_{s}$, we consider in turn the EPFs after staying or exiting. Case 1: $\mathsf{P}_1$ is recommending $\mathsf{P}_2$ to stay. 
For incentive compatibility, it needs to \textit{promise} $\mathsf{P}_2$ a payoff of at least $0$ under $T(\text{stay};s)=s'$. 
Thus, we \textbf{left-truncate} the regions of the EPF at $s'$ which violate this promise, leaving behind the blue segment (Figure~\ref{fig:demo-2}), which represents the payoffs at $s'$ that are \textit{enforceable} by $\mathsf{P}_1$. Case 2: $\mathsf{P}_1$ is recommending $\mathsf{P}_2$ to exit. To discourage $\mathsf{P}_2$ from staying, it commits to the grim strategy at $s'$ if $\mathsf{P}_2$ chooses to stay instead, yielding $\mathsf{P}_2$ a payoff of $-1 \leq k_2=0$. Hence, no truncation is needed and the set of enforceable payoffs is the (degenerate) blue line segment (Figure~\ref{fig:demo-3}). Finally, to recover $U_{s}$, observe that we can achieve payoffs on any line segment connecting point across the EPFs of $s$'s children.
This union of points on such lines (ignoring those leader-dominated) is given by the \textbf{upper concave envelope} of the blue segments in Figure~\ref{fig:demo-2} and \ref{fig:demo-3}; this removes $\{(0, 0)\}$, giving the EPF in Figure~\ref{fig:demo-4}.

More generally, let $g_1$ and $g_2$ be functions such that $g_j : \mathbb{R} \mapsto \mathbb{R} \cup \{ -\infty \}$. We denote by $g_1 \bigwedge g_2$ their upper concave envelope, i.e., $\inf \{ h(\mu) \mid h \text{ is concave and } h \geq \max \{ g_1, g_2 \} \text{ over } \mathbb{R} \}$. Since $\bigwedge$ is associative and commutative, we use as shorthand $\bigwedge_{\{ \cdot \}}$ when applying $\bigwedge$ repeatedly over a finite set of functions. In addition, we denote $g \triangleright t $ as the left-truncation of the $g$ with threshold $t \in \mathbb{R}$, i.e., $[g \triangleright t](\mu) = g(\mu)$ if $\mu \geq t$ and $-\infty$ otherwise. Note that both $\bigwedge$ and $\triangleright$ are closed over concave functions. For any $s \in \mathcal{S}$, its EPF $U_s$ can be concisely written in terms of its children EPF $U_{s'}$ (where $s'\in\mathcal{C}(s)$) using $\bigwedge$, $\triangleright$ and $\tau(s')$. 
\begin{align}
    U_s (\mu) &= \begin{cases}
    \left[ \bigwedge_{s' \in \mathcal{C}(s)}U_{s'} \right] (\mu) & \text{ if } s \in \mathcal{S}_1 \\
    \left[ \bigwedge_{ s' \in \mathcal{C}(s)} U_{s'} \triangleright \tau(s') \right] (\mu) & \text{ if } s \in \mathcal{S}_2
    \end{cases},
    \label{eq:sefce_recurse}
\end{align}
which we apply in a bottom-up fashion to complete Phase 1.

\paragraph{Phase 2: Extracting Strategies from EPF.} Once $U_s$ has been computed for all $s \in \mathcal{S}$, we can recover the optimal strategy $\pi_1$ by applying one-step lookahead starting from the root. First, we extract $(\optval_2, \optval_1)$, the coordinates of the maximum point in $U_{\text{root}}$, which contain payoffs under the optimal $\pi$. Here, this is $(0, 4.5)$. We initialize $\mu_2=\optval_2$, which represents $\mathsf{P}_1$'s promised payoff to $\mathsf{P}_2$ at the current state $s$. Next, we traverse $\game$ depth-first. By construction, $U_s (\mu_2) > -\infty$ and the point $(\mu_2, U_s (\mu_2))$ is the convex combination of either $1$ or $2$ points belonging to its children EPFs. The mixing factors correspond to the optimal strategy $\pi(a; s)$. If there are $2$ distinct children $s', s''$ with mixing factor $\alpha', \alpha''$, we repeat this process for $s', s''$ with $\mu'_2 = \mu_2 / \alpha', \mu''_2 = \mu_2 / \alpha''$, otherwise we repeat the process for $s'$ and $\mu'_2 = \mu_2$. For our example, we start at $s$, $\mu_2=0$, which was obtained by $\mathsf{P}_2$ playing `stay' exclusively, so we keep $\mu_2$ and move to $s'$. At $s'$, $\mu=0$ by mixing uniformly, which gives us the result in Section~\ref{sec:background}.


\begin{theorem}[\cite{bovsansky2017computation,DBLP:journals/corr/BosanskyBHMS15}]\label{thm:plc}
(i) $U_s$ is piecewise linear concave with number of knots\footnote{Knots are where the slope of the EPF changes.} no greater than the number of leaves beneath $s$. (ii) Using backward induction, SEFCEs can be computed in polynomial time (in $|\mathcal{S}|$) even in games with chance. EPFs continue to be piecewise linear concave.
\end{theorem}

\paragraph{Markovian Property.} 
Just like state values $v_s$ in zero-sum games, we can replace any internal vertex $s$ in $\game$ with its EPF while not affecting the optimal strategy in all other branches of the game. This can done by adding a single leader vertex with actions leading to terminal states with payoffs corresponding to the knots of $U_s$. Since $U_s$ is obtained via backward induction, it only depends on states beneath $s$. In fact, if two games $\game$ and $\game'$ (which could be equal to $\game$) shared a common subgame rooted in $s$ and $s'$ respectively, we could reuse the $U_s$ found in $\game$ for $U_{s'}$ in $\game'$. This observation underpins the inspiration for our work---if $s$ and $s'$ are similar in some features, then $U_s$ and $U_{s'}$ are likely similar and it should be possible to \textit{learn} and \textit{generalize} EPFs over states. 

\section{Challenges in Applying FA to EPF}
\label{sec:challenges}
We now return to our original problem of applying FA to find SEFCE.
Our idea, outlined in Algorithm~\ref{alg:pipeline} and \ref{alg:bellman} is a straightforward extension of FVI. 
Suppose each state has features $f(s)$---in the simplest case this could be a state's history. We design a neural network $E_{\phi}(f)$ parameterized by $\phi$. This network maps state features $f(s)$ to some representation of $\tilde{U}_s$, the approximated EPFs. 
To achieve a good approximation, we optimize $\phi$ by minimizing an appropriate Bellman-like loss (over EPFs) based on Equation~\eqref{eq:sefce_recurse} while using our approximation $\tilde{U}_s$ in lieu of $U_s$. Despite its simplicity, there remain several design considerations.
\begin{figure}[t]
\centering
\begin{minipage}{0.45 \textwidth}
\begin{algorithm}[H]
\caption{Training Pipeline}
\label{alg:pipeline}
\begin{algorithmic}[1] 
\STATE {Sample trajectory $s_{\text{new}}^{(1)}, \dots, s_{\text{new}}^{(t)}$}
\STATE {Update replay buffer $\mathcal{B}$ with $s^{(1)}, \dots, s^{(t)}$}
\FOR {$i \in \{1,\dots,t \} $}
    \STATE{Sample batch $S=\{ s^{(1)}, \dots s^{(n)} \}  \subseteq \mathcal{B}$ }
    \STATE{$\ell \leftarrow \textsc{ComputeLoss}(S; E_{\phi})$}
    \STATE{Update $\phi$ using $\partial \ell/\partial\phi$}
\ENDFOR
\end{algorithmic}
\end{algorithm}
\end{minipage}
\begin{minipage}{0.45 \textwidth}
\begin{algorithm}[H]
\caption{\textsc{ComputeLoss}$(S; E_{\phi})$}
\label{alg:bellman}
\begin{algorithmic}[1] 
\FOR {$i \in \{ 1\dots n \} $}
\STATE { $\tilde{U}_{s^{(i)}} \leftarrow E_{\phi}(f(s^{(i)}))$}
\STATE {$\tilde{U}_{s_{\text{next}}^{(j)}} \leftarrow E_{\phi}(f(s_{\text{next}}^{(j)})) \quad \text{for all } s_{\text{next}}^{(j)} \in \mathcal{C}(s^{(i)})$}
\STATE {Compute $\tilde{U}^{\text{target}}_{s^{(i)}}$ using Equation~\eqref{eq:sefce_recurse} and $\{ \tilde{U}_{s_{\text{next}}^{(j)}} \}$} \label{line:lookahead}
\ENDFOR
\RETURN {$\sum_{i} L(\tilde{U}_{s^{(i)}}, \tilde{U}^{\text{target}}_{s^{(i)}})$} 
\end{algorithmic}
\end{algorithm}
\end{minipage}
\end{figure}

\paragraph{EPFs are the `right' object to learn.}
Unlike state values, representing an exact EPF at a state $s$ could require more than constant memory since the number of knots could be linear in the number of leaves underneath it (Theorem~\ref{thm:plc}). Can we get away with summarizing a state with a scalar or a small vector? Unfortunately, any `lossless summary' which enjoys the Markovian property necessarily encapsulates the EPF. To see why, consider the class of games $\game_{k}$ in Figure~\ref{fig:basic_eg} with $k_1=-2$ and $k=k_2 \in [-1, 1]$. 
The optimal leader payoff for any $\game_k$ is $\frac{9-11k}{2}$, which is precisely $U_{s'}(k)$ (Figure~\ref{fig:demo-2}).
Now consider any lossless summary for $s'$ and use it to solve \textit{every} $\game_{k}$. The resultant optimal leader payoffs can recover $U_{s'}(\mu_2)$ between $\mu_2 \in [-1,1]$. This implies that no lossless summary more compact than the EPF exists.


\paragraph{Unfulfillable Promises Arising from FA Error.}
Consider the game in Figure~\ref{fig:promise} with $k_1=-10$, $k_2=-1$. The exact $U_{s'}$ is the line segment combining the points $(-1, 10)$ and $(1-\epsilon, -1)$, shown in green in Figure~\ref{fig:unfilfillable}. However, let us suppose that due to function approximation we instead learned the blue line segment containing $(-1, 10)$ and $(1, -1)$. Performing Phase 2 using $\tilde{U}$, the policy extracted at $s'$ is once again the uniform policy and requires us to promise the follower a utility of $1$ in $s''$. However, achieving a payoff of $1$ is \textit{impossible} regardless of how much the leader is willing to sacrifice, since the maximum outcome under $s''$ is $1-\epsilon$. Since this is an \textit{unfulfillable promise}, the follower's best responds by exiting in $s$, which gives the leader a payoff of $-10$. In general, unfulfillable promises due to small FA error can lead to arbitrarily low payoffs. 
In fact, one could argue that $\tilde{U}$ does not even \textit{define} a valid policy.

\paragraph{Costly Promises.}
Consider the case where $k_1=-30, k_2=1$ while keeping $\tilde{U}_{s'}$ the same. Here, the promise of $1$ at $s''$ \textit{is} fulfillable, but involves incurring a cost of $-30$, which is even lower than having follower staying (Figure~\ref{fig:costly}). In general, this problem of \textit{costly promises} stems from the EPF being wrongly estimated, even for a small range of $\mu_2$. We can see how costly promises arise even from small $\epsilon$ is. The underlying issue is that in general, $U_s$ can have large Lipschitz constants (e.g., proportionate to $\left(\max_s r_1(s)-\min_s r_1(s) \right)/(\min|r_2(s)-r_2(s)|)$). The existence of costly payoffs rules out EPF representations based on discretizing the space of $\mu_2$, since small errors incurred by discretization could lead to huge drops in performance.

\begin{figure}[t]
\centering
\begin{subfigure}[b]{0.48\linewidth}
    \centering
    \begin{tikzpicture}[scale=1,font=\footnotesize]
\tikzstyle{leader}=[regular polygon,regular polygon sides=50,draw,inner sep=1.2,minimum size=12];
\tikzstyle{follower}=[regular polygon,regular polygon sides=50, rotate=180,draw,inner sep=1.2, minimum size=12, fill=gray!65];
\tikzstyle{terminal}=[draw,inner sep=1.2];
\tikzstyle{level 1}=[level distance=3mm,sibling distance=20mm]
\tikzstyle{level 2}=[level distance=3mm,sibling distance=12mm]
\tikzstyle{level 3}=[level distance=6.5mm,sibling distance=15mm]
\node(0)[follower, label=center:{$s$}]{}
    child{node(1)[leader, label=center:{$s'$}]{}
        child{node(3)[terminal, label=below:{$(10,-1)$}]{}
        }
        child{node(4)[leader, label=center:{$s''$}]{}
            child{node(10)[terminal, label=below:{$(k_1,k_2)$}]{}}
            child{node(11)[terminal, label=below:{$(-1,1-\epsilon)$}]{}}
        }
        edge from parent node[above left]{}
    }
    child{node(2)[terminal, label=below:{$(-10,0)$}]{}
    edge from parent node[above right]{}
    };
\end{tikzpicture}
    \subcaption{Game tree to illustrate unfulfillable and costly promises.}
    \label{fig:promise}
    \end{subfigure}
\begin{subfigure}[b]{0.48\linewidth}
    \centering
    \begin{tikzpicture}[scale=1,font=\footnotesize]
\tikzstyle{leader}=[regular polygon,regular polygon sides=50,draw,inner sep=1.2,minimum size=12];
\tikzstyle{follower}=[regular polygon,regular polygon sides=50, rotate=180,draw,inner sep=1.2, minimum size=12, fill=gray!65];
\tikzstyle{terminal}=[draw,inner sep=1.2];
\tikzstyle{level 1}=[level distance=7mm,sibling distance=20mm]
\tikzstyle{level 2}=[level distance=3mm,sibling distance=12mm]

\node(0)[follower, label=center:{$s$}]{}
    child{node(1)[leader, label=center:{$s'$}]{}
        child{node(3)[terminal, label=below:{$(0,0)$}]{}
        }
        child{node(4)[terminal, label=below:{$(-1,-1)$}]{}
        }
        edge from parent node[above left]{refuse}
    }
    child{node(2)[terminal, label=below:{$(q_1,-q_2)$}]{}
    edge from parent node[above right]{accede}
    };
\end{tikzpicture}
    \subcaption{Stage game used in the \textsc{Tantrum} game.}
    \label{fig:threat_game}
\end{subfigure}
\caption{Games used in Sections~\ref{sec:challenges} and \ref{sec:experiments}. Leader \tikz\draw[black] (0,0) circle (.7ex);, follower  \tikz\draw[black,fill=gray] (0,0) circle (.7ex); and leaf $\square$ states are vertices and edges are actions.}
\end{figure}

\begin{figure}[t]
\centering
\begin{subfigure}{.48\linewidth}
    \centering
    \includegraphics[width=\textwidth]{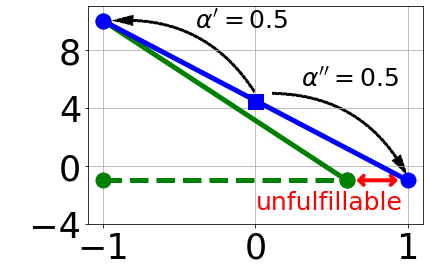}
    \subcaption{Unfulfillable promise.}
    \label{fig:unfilfillable}
\end{subfigure}
\begin{subfigure}{.48\linewidth}
    \centering
    \includegraphics[width=\textwidth]{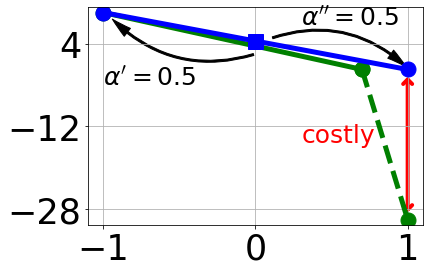}
    \subcaption{Costly promise.}
    \label{fig:costly}
\end{subfigure}. \caption{EPFs for (a) unfulfillable promises and (b) costly promises. Blue lines are estimated EPFs $\tilde{U}_{s'}$, solid and dotted green lines are true EPFs $U_{s'}$, $U_{s''}$. In both cases, FA error leads us to believe that the payoff given by the blue square at $(0, 4.5)$ can be achieved by mixing the endpoints of $\tilde{U}_{s'}$ with probability $\alpha'=\alpha''=0.5$ (black curves).}
\end{figure}

\section{FA of EPF with Performance Guarantees}
\label{sec:our_method}
We now design our method using the insights from Section~\ref{sec:challenges}. We learn EPFs without relying on discretization over $\mathsf{P}_2$ payoffs $\mu_2$. Unfulfillable promises are avoided entirely by ensuring that the set of $\mu_2$ where $\tilde{U}_s(\mu_2) > -\infty$ lies within some known set of achievable $\mathsf{P}_2$ payoffs, while costly promises are mitigated by suitable loss functions.
\paragraph{Representing EPFs using Neural Networks.}
Our proposed network architecture represents EPFs by a small set of $m \geq 2$ points $P_\phi(s) = \{(x_j, y_j)\}$, for $j \in [m]$. Here, $m$ is a hyperparameter trading off complexity of the neural network $E_\phi$ with its representation power. The approximated EPF $\tilde{U}_s$ is the linear interpolation of these $m$ points; and $\tilde{U}_s = -\infty$ if $\mu_2 > \max_j x_j$ or $\mu_2 < \min_j x_j$. For now, we make the assumption that follower payoffs under the altruistic and grim strategy ($\UBV(s)$ and $\LBV(s)$) are known \textit{exactly} for all states. Through the architecture of $E_\phi$ that for all $j \in [m]$, we have $\LBV(s) \leq x_j \leq \UBV(s)$. As we will see, this helps avoid unfulfillable promises and allows for convenient loss functions. 

Concretely, our network $E_\phi(f(s); \LBV(s), \UBV(s))$ takes in as inputs state features $f(s)$, lower and upper bounds $\LBV(s) \leq \UBV(s)$ and outputs a matrix in $\mathbb{R}^{m \times 2}$ representing $\{(x_j, y_j)\}$ where $x_1 = \LBV(s)$ and $x_m = \UBV(s)$. For simplicity, we use a multilayer feedforward network with depth $d$, width $w$ and ReLU activations for each layer. Serious applications should utilize domain specific architectures. Denoting the output of the last fully connected layer by $h^{(d)}(f(s)) \in \mathbb{R}^w$, for $j \in \{ 2 \dots m-1 \}$ and $k \in [m]$ we set
\begin{align*}
x_j &= \sigma \left( z_{x, j}^T h^{(d)}(f(s)) + b_{x, j} \right) \cdot \left({\UBV(s)-\LBV(s) }\right) + \LBV(s),
\\ 
y_k & = z_{y, k}^T h^{(d)}(f(s)) + b_{y, k},
\end{align*}
and $x_1 = \LBV(s)$ and $x_m = \UBV(s)$, where $\sigma(x)=1/(1+\exp(-x))$. Here, $z_{x,j}, z_{y,k} \in \mathbb{R}^w$ and $b_{x,j}, b_{y,k} \in \mathbb{R}$ are weights and biases, which alongside the parameters from feedforward network form the network parameters $\phi$ to be optimized. Since $\tilde{U}_s$ is represented by its knots (given by $P_\phi(s)$), $\bigwedge$ and consequently, \eqref{eq:sefce_recurse} may be performed \textit{explicitly} and \textit{efficiently}, returning an entire EPF represented by its knots (as opposed to the EPF evaluated at a single point). This is crucial, since the computation is performed every state every iteration (Line~\ref{line:lookahead} of Algorithm~\ref{alg:bellman}).

\paragraph{Loss Functions for Learning EPFs.} 
Given $2$ EPFs $\tilde{U}_s, \tilde{U}'_s$ 
we minimize the following loss to mitigate costly promises,
\begin{align*}
L_{\infty}(\tilde{U}_s, \tilde{U}'_s) = \max_{\mu_2} | \tilde{U}_s(\mu_2) - \tilde{U}'_s(\mu_2)|.
\end{align*}
$L_\infty$ was chosen specifically to incur a large loss if the approximation is wildly inaccurate in a small range of $\mu_2$ (e.g., Figure~\ref{fig:costly}). Achieving a small loss requires that $\tilde{U}_s(\mu_2)$ approximates $\tilde{U}'_s(\mu_2))$ well for all $\mu_2$. This design decision is particularly important. For example, contrast $L_\infty$ with another intuitive loss $L_2 (\tilde{U}_s, \tilde{U}'_s) = \int_{\mu_2} ( \tilde{U}_s(\mu_2) - \tilde{U}'_s(\mu_2))^2 d\mu_2$. Observe that $L_2$ is exceedingly small in the example of Figure~\ref{fig:costly} --- in fact, when $\epsilon$ is small enough leads to almost no loss, even though the policy as discussed in Section~\ref{sec:challenges} is highly suboptimal. This phenomena leads to costly promises, which was indeed observed in our tests. 
\paragraph{Our Guarantees. }
Any learned $\tilde{U}$ implicitly defines a policy $\tilde{\pi}$ by one-step lookahead using Equation~\eqref{eq:sefce_recurse} and the method described in Phase 2 (Section~\ref{sec:review}). Extracting $\tilde{\pi}$ need not be done offline for all $s \in \mathcal{S}$; in fact, when $\game$ is too large it is necessary that we only extract $\tilde{\pi}(\cdot; s)$ on-demand. Nonetheless, $\tilde{\pi}$ enjoys some important properties.

\begin{theorem}[Incentive Compatibility]
For any policy $\tilde{\pi}$ obtained using our method, any $s \in \mathcal{S}_2$ and $a \in \mathcal{A}(s)$, we have $\tilde{\pi}_s(a;s) > 0 \Longrightarrow R_2(T(a; s); \tilde{\pi}) \geq \tau(T(a; s))$.
\label{thm:incentive}
\end{theorem}
\begin{theorem}[FA Error]
If $L_{\infty}(\tilde{U}_s, \tilde{U}^{\text{target}}_s) \leq \epsilon$ for all $s \in \mathcal{S}$, then $|R_1(\tilde{\pi}) - R_1(\pi^*)| = \mathcal{O}(D \epsilon)$ where $D$ is the depth of $\game$ and $\pi^*$ is the optimal strategy.
\label{thm:FA_guarantee}
\end{theorem}
Here, $T(a; s)$ is transition function (Section~\ref{sec:background}). Recall from Section~\ref{sec:background} that for $\pi$ to be an optimal SEFCE, we require (i) incentive compatibility and (ii) $R_1(\pi)$ to be maximized. 
Theorems~\ref{thm:incentive} and \ref{thm:FA_guarantee} illustrate how our approach disentangles these criteria. Theorem~\ref{thm:incentive} guarantees that $\mathsf{P}_2$ will always be incentivized to follow $\mathsf{P}_1$'s recommendations, i.e., there will be no unexpected outcomes arising from unfulfillable promises. Crucially, this is a hard constraint which is satisfied solely due to our choice of network architecture, which ensures that $\tilde{U}_s(\mu_2)=-\infty$ when $\mu_2 > \UBV_s$ for \textit{any} $\tilde{\pi}$ obtained from $\tilde{U}$. Conversely, Theorem~\ref{thm:FA_guarantee} shows that the goal of maximizing $R_1$ \textit{subject to incentive compatibility} is achieved by attaining a small FA error across all states.
This distinction is important. Most notably, incentive compatibility is no longer dependent on convergence during training. This \textit{explicit} guarantee stands in contrast with methods employing self-play reinforcement learning agents; there, incentive compatibility follows \textit{implicitly} from the apparent convergence of a player's strategy. This guarantee has practical implications, for example, evaluating the quality of $\tilde{\pi}$ can be done by estimating $R_1(\tilde{\pi})$ based on sampled trajectories, while implicit guarantees requires incentive compatibility to be demonstrated using some approximate best-response oracle and usually involves expensive training of a RL agent.

The primary limitation of our method is when $\UBV$ and $\LBV$ (and hence $\tau$) are not known exactly. As it turns out, we can instead use upper and lower approximations while still retaining incentive compatibility. Let $\tilde{\pi}_1^{\text{grim}}$ be an \textit{approximate grim strategy}. 
Define $\ALBV(s)$ to be the expected follower payoffs at $s$ when faced best-responding to $\tilde{\pi}_1^{\text{grim}}$, i.e., $R_2(s; \tilde{\pi}_1^{\text{grim}}, \pi_2)$, where $\pi_2 \in BRS_2(\tilde{\pi}_1^{\text{grim}})$. 
Following Definition~\ref{def:approx_min_req_incentive}, the \textit{approximate minimum required incentive} is $\tilde{\tau}(s')=\max_{s^! \in \mathcal{C}(s);s^! \neq s'} \ALBV(s^!)$ for all $s \in \mathcal{S}_2$, $s'\in\mathcal{C}(s)$.
Similarly, let $\tilde{\pi}^{\text{alt}}$ be an \textit{approximate joint altruistic strategy} and its resultant internal payoffs in each state be $\AUBV(s)$. 

Under the mild assumption that $\tilde{\pi}^{\text{alt}}$ always benefits $\mathsf{P}_2$ more than the $\tilde{\pi}_1^{\text{grim}}$, i.e., $\AUBV(s) \geq \ALBV(s)$ for all $s$, we can replace the $\tau, \LBV$ and $\UBV$ with $\tilde{\tau}, \ALBV$ and $\AUBV$ and maintain incentive compatibility (Theorem~\ref{thm:incentive}). The intuition is straightforward: if $\mathsf{P}_2$'s threats are `good enough', parts of the EPF will still be enforceable. Furthermore, promises will always be fulfillable since  EPFs domains are now limited to be no greater than $\AUBV(s)$, which we know can be achieved by definition. Unfortunately, Theorem~\ref{thm:FA_guarantee} no longer holds, not even in terms of $\max_{s}|\tilde{\tau}(s)-\tau(s)|$. This is again due to the large Lipschitz constants of $U_s$. However, we have the weaker guarantee (whose proof follows that of Theorem~\ref{thm:FA_guarantee}) that performance is close to that predicted at the root.
\begin{theorem}[FA Error with Weaker Bounds]
If $L_{\infty}(\tilde{U}_s, \tilde{U}^{\text{target}}_s) \leq \epsilon$ for all $s \in \mathcal{S}$, then $|R_1(\tilde{\pi}) - \widetilde{\mathsf{OPT}}_2| = \mathcal{O}(D \epsilon)$ where $D$ is the depth of $\game$ and $\widetilde{\mathsf{OPT}}_2 = \max_{\mu_2} \tilde{U}_{\text{root}}(\mu_2)$.
\label{thm:FA_guarantee_approx}
\end{theorem}

\textbf{Remark.} The key technical difficulty here is finding $\AUBV$. In our experiments, $\tilde{\pi}_1^{\text{grim}}$ can be found analytically. In general large games, we can approximate $\tilde{\pi}_1^{\text{grim}}$, $\AUBV$ by searching over $\mathcal{S}_2$, 
but use heuristics when expanding nodes in $\mathcal{S}_1$.

\paragraph{Implementation Details.} 
(i) We use several techniques typically used to stabilize training such as target networks \cite{arulkumaran2017deep,mnih2015human} and prioritized experience replay \cite{schaul2015prioritized}. 
(ii) In practice, instead of $L_\infty$, we found it easier to train a loss based on the sum of the squared distances at the $x$-coordinate of the knots in $\tilde{U}_s$ and $\tilde{U}_s'$, i.e., $L=\sum_{\mu_2 \in \{\text{knots}\}} [ \tilde{U}_s(\mu_2) - \tilde{U}'_{s}(\mu_2)]^2$. Since $L$ upper bounds $L_\infty^2$, using it also avoids costly promises and allows us to enjoys a similar FA guarantee. (iii) If $\game$ has a branching factor of $\beta$, then \eqref{eq:sefce_recurse} in Algorithm~\ref{alg:bellman} can be executed in $\mathcal{O}(\beta m)$ time. In practice, we use a brute force method better suited for batch GPU operations which runs in $\mathcal{O}((\beta m)^3)$.
(iv) We train using only the decreasing portions of $\tilde{U}_s$. This does not lead to any loss in performance since payoffs in the increasing portion of an $\tilde{U}_s$ are Pareto dominated. 
We do not want to `waste' knots on learning the meaningless increasing portion.
(v) Training trajectories were obtained by taking actions uniformly at random. 
Specifics for all implementation details are in the Appendix.

\section{Experiments}
\label{sec:experiments}
\begin{figure*}[t]
\begin{subfigure}[b]{0.3\linewidth}
\setlength{\tabcolsep}{3.5pt}
\begin{tabular}[c]{ccccc}
        & \#& $\overline{\Delta_{\text{OPT}}}$&$\overline{\Delta_{\text{SP}}}$ & $\overline{\Delta_{\text{non}}}$ \\
        \hline \hline
        \textsc{RC} & 10 & -.0247 & .200& .265\\
        \hline
        \textsc{Tantrum} & 5 & -.0262 & 8.89 & N/A\\
        \hline
        \textsc{RC}+ & 3 & N/A  & N/A & .421
\end{tabular}
\subcaption{Results for fixed parameter games}
\label{fig:fix_param_results}
\end{subfigure}
\begin{subfigure}[b]{0.22 \linewidth}
    \includegraphics[width=\linewidth]{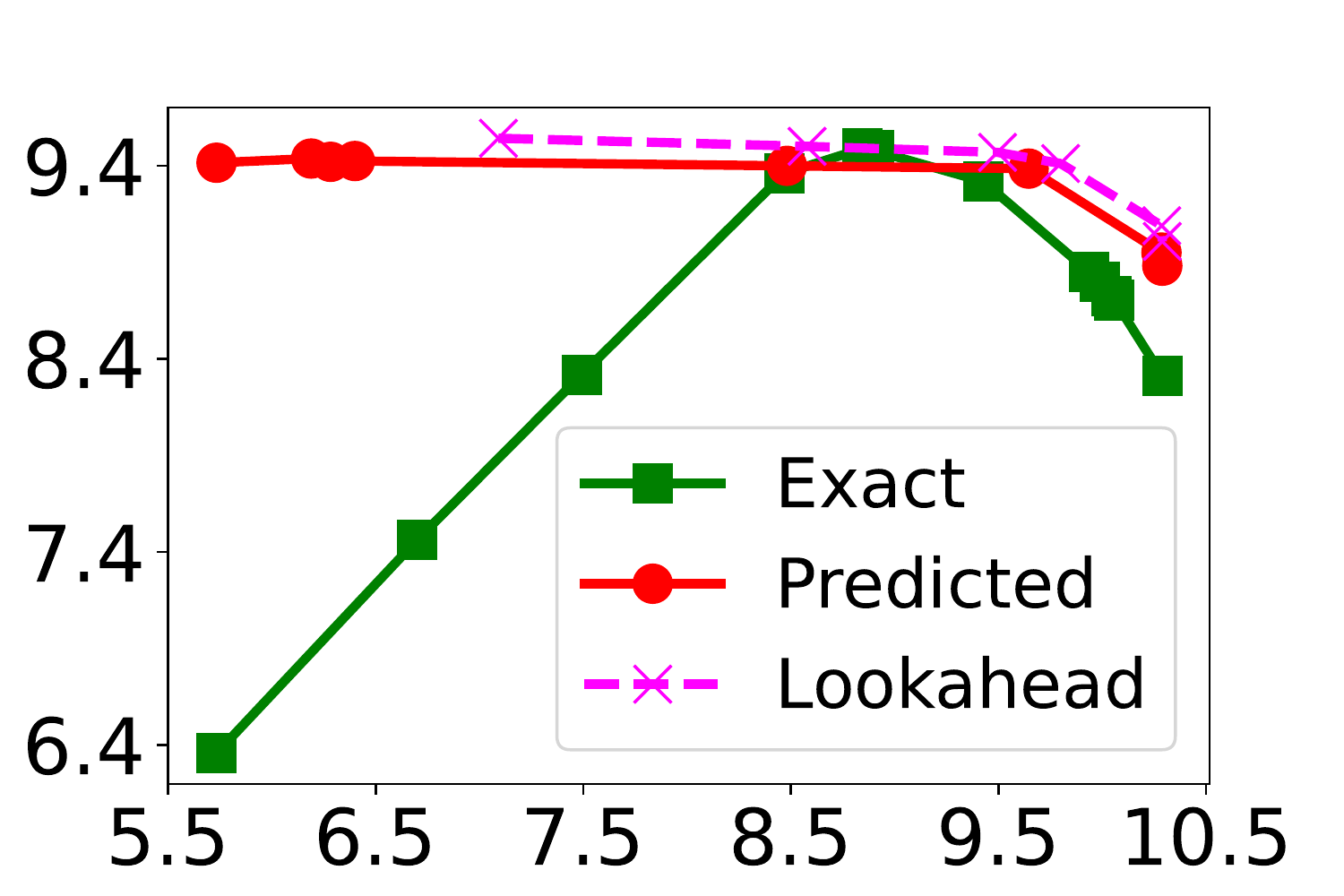}
    \subcaption{EPF after 100k epochs}
    \label{fig:epoch-comparison-100k}
    \end{subfigure}
    \begin{subfigure}[b]{0.22 \linewidth}
    \includegraphics[width=\linewidth]{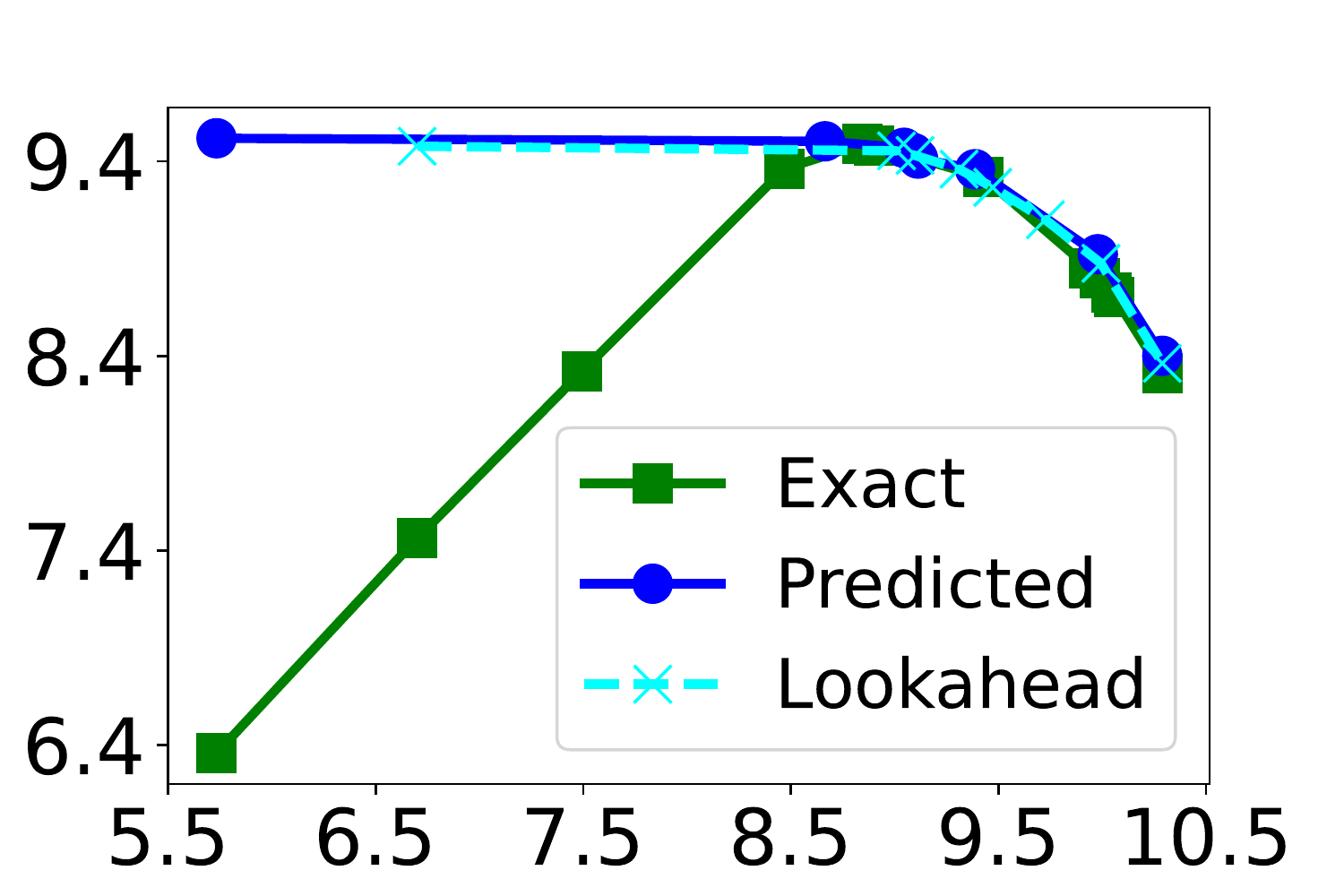}
    \subcaption{EPF after 2M epochs}
    \label{fig:epoch-comparison-2M}
    \end{subfigure}
    \begin{subfigure}[b]{0.22 \linewidth}
    \includegraphics[width=\linewidth]{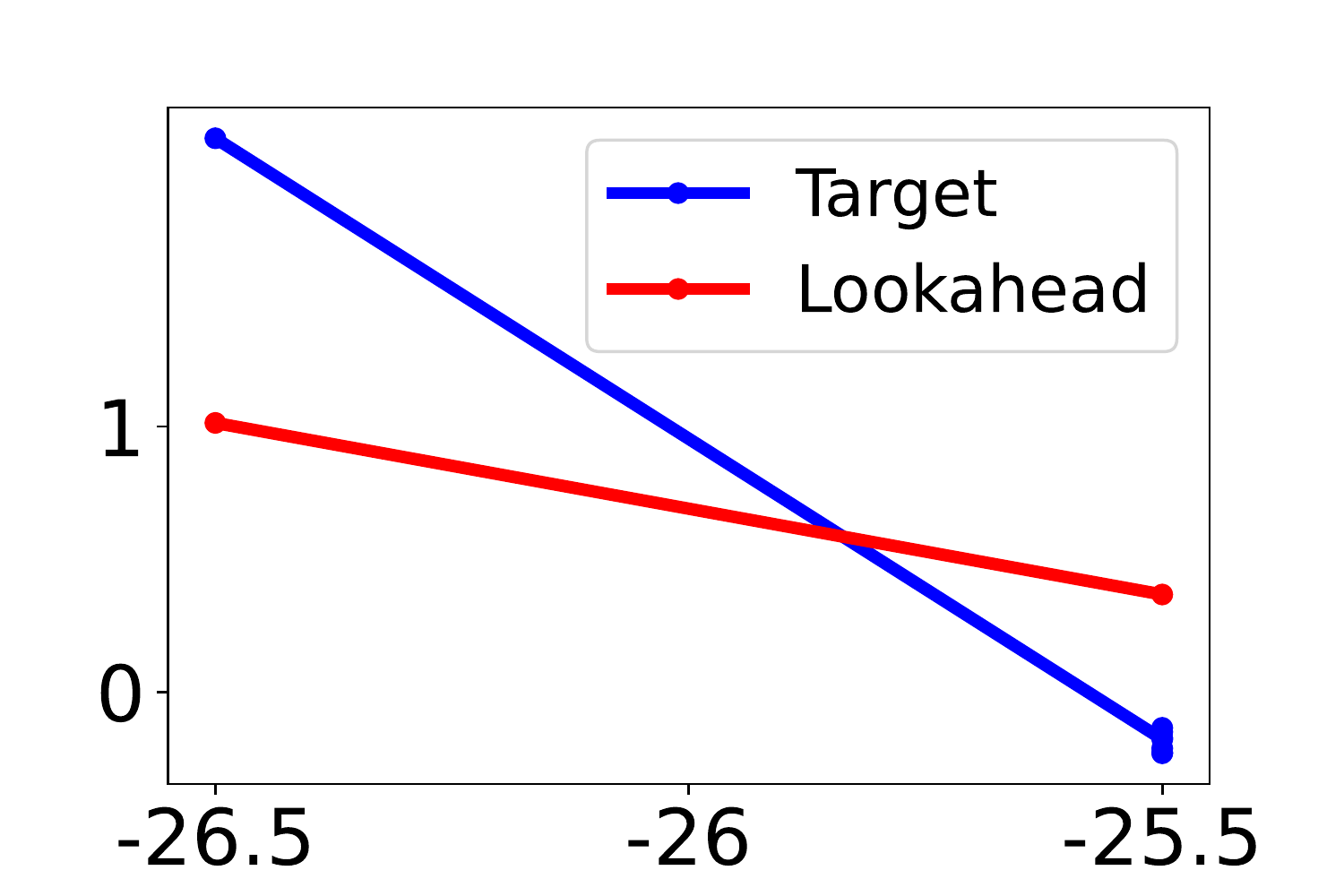}
    \subcaption{Failure case}
    \label{fig:threat-large-fa}
    \end{subfigure}
\caption{(a) Results for games with fixed parameters averaged over \# specifies \# trials. $\overline{\Delta_{\text{OPT}}}$, $\overline{\Delta_{\text{SP}}}$, and $\overline{\Delta_{\text{non}}}$ is the average difference between our method and the optimal SEFCE, subgame perfect Nash, and non-strategic leader commitment.
(b)-(c) Learned EPFs at the root for $\textsc{RC}$. (d) A failure case in \textsc{Tantrum}, even though learned policies are still near-optimal.
}
\label{a}
\end{figure*}
We focused on the following two synthetic games. Game details and experiment environments are in the Appendix. Code is at https://github.com/lingchunkai/learn-epf-sefce.

\paragraph{Tantrum.} \textsc{Tantrum} is the game in Figure~\ref{fig:threat_game} repeated $n$ times, with $q_1 > 0, q_2 \geq 1$, and rewards accumulated over stages. The only way $\mathsf{P}_1$ can get positive payoffs is by threatening to throw a trantrum with the mutually destructive $(-1, -1)$ outcome. Since $q_2 > 1$, $\mathsf{P}_2$ has to use threats spanning over stages to sufficiently entice $\mathsf{P}_2$ to accede. Even though \textsc{Trantrum} has $\mathcal{O}(3^n)$ leaves, it is clear that the grim (resp. altruistic) strategy is to throw (resp. not throw) a tantrum at every step. Hence $\UBV$ and $\LBV$ are known even when $n$ is large, making \textsc{Tantrum} a good testbed. The raw features $f(s)$ is a 5-dimensional vector, the first 3 are the occurrences count of outcomes for previous stages, and the last 2 being a one-hot vector indicating the current state.

\paragraph{Resource Collection.} \textsc{RC} is played on a $J \times J$ grid with a time horizon $n$. Each cell contains varying quantities of 2 different resources $\textsf{r}_1(x, y), \textsf{r}_2(x, y) \geq 0$, both of which are collected (at most once) by either players entering. Players begin in the center and alternately choose to either move to an adjacent cell or stay put. Each $\mathsf{P}_i$ is only interested in resource $i$, and players agree to pool together resources when the game ends. \textsc{RC} gives $\mathsf{P}_1$ the opportunity to threaten $\mathsf{P}_2$ with going `on strike' if $\mathsf{P}_2$ does not move to the cells that $\mathsf{P}_1$ recommends. \textsc{RC} has approximately $\mathcal{O}({25^n})$ leaves. The grim strategy is for $\mathsf{P}_1$ to stay put. However, unlike \textsc{Tantrum}, computing $\LBV$ and $\UBV$ still requires search (at least for $\mathsf{P}_2$) at each state, which is still computationally expensive. We use as features (a) one-hot vector showing past visited locations, (b) the current coordinates of each player and whose turn it is (c) the amount of each resource collected, and (d) the number of rounds remaining.

\subsection{Experimental Setup}

\paragraph{Games with Fixed Parameters.} We run 3 sub-experiments. [\textbf{\textsc{RC}}] We experimented with \textsc{RC} with $J=7, n=4$ over 10 different games. Rewards $\textsf{r}_i$ were generated using a log-Gaussian process over $(x, y)$ to simulate spatial correlations (details in Appendix). 
We also report the payoffs from a `non-strategic' $\mathsf{P}_1$ which optimizes only for resources it collects, while letting $\mathsf{P}_2$ best respond. [\textbf{\textsc{Tantrum}}] We ran \textsc{Tantrum} with $n=25$, $q_1=1$ and $q_2$ chosen randomly. These games have $>1e12$ states; however, we can still obtain the optimal strategy due to the special structure of the game (note the subgame perfect equilibrium gives $\mathsf{P}_1$ zero payoff). [\textbf{\textsc{RC}+}] We ran $\textsc{RC}$ with $J=9$, $n=6$. Since $\game$ is large, we use approximates ($\tilde{\tau}$, $\AUBV$, $\ALBV$) obtained from $\tilde{\pi}_1^{\text{grim}}$ and $\tilde{\pi}^{\text{alt}}$. $\tilde{\pi}_1^{\text{grim}}$ is for $\mathsf{P}_2$ to stay put, while $\ALBV$ is obtained by applying search \textit{online} (i.e., when $s$ appears in training) for $\mathsf{P}_2$ starting from $s$. Thus $\tilde{\tau}(s)$ can also be computed online from $\ALBV$. $\tilde{\pi}^{\text{alt}}$ is obtained by running exact search to a depth of $4$ (counted from the root) and then switching to a greedy algorithm. On the rare occasion that $\AUBV(s) < \ALBV(s)$, we set $\AUBV(s) \leftarrow \ALBV(s)$. We report results in Figure~\ref{fig:fix_param_results}, which show the \textit{difference} between $\mathsf{P}_1$'s payoff for our method and (i) the optimal SEFCE, (ii) the subgame perfect Nash, and (iii) the non-strategic leader commitment. 

\paragraph{Featurized \textsc{Tantrum}.} We allow $q_1, q_2$ to \textit{vary} between stages of $\game$, giving vectors $\textbf{q}_i \in [1,\infty]^n$. Each trajectory uses different $\textbf{q}_i$, which we append as features to our network, alongside the payoffs already collected for each player. For training, we draw i.i.d. samples of $\textbf{q}_i^j \sim \exp(1) + 1$. The evaluation metric is $\kappa = R_1(\tilde{\pi})/\textsf{OPT}$, i.e., the ratio of $\mathsf{P}_1$'s payoffs under $\tilde{\pi}$ compared to the optimal $\pi$. For each $n$, we test on $100$ $\textbf{q}$-vectors not seen during training and compare their $\kappa$ against a `greedy' strategy which recommends $\mathsf{P}_2$ to accede as long as there are sufficient threats in the remainder of the game for $\mathsf{P}_1$ (details in Appendix). We also stress test $\tilde{\pi}$ on a different \textit{test} distribution $\hat{\textbf{q}}_i^j \sim \exp(1) + 4$. We report results in Figure~\ref{fig:featurized-results} and \ref{fig:featurized-results-2}.

\begin{figure}[t]
\centering
\begin{subtable}[c]{0.40\linewidth}
\setlength{\tabcolsep}{3.5pt}
\begin{tabular}{cccc}
        $n$ & $\overline{\kappa}$ & grd-$\overline{\kappa}$ & str-$\overline{\kappa}$ \\
        \hline \hline
        5 & .993 & .828 & .997\\
        \hline
        6 & .982 & .773 & .982\\
        \hline
        7 & .968 & .778 & .921 \\
        \hline
        10 & .938 & .775 & .898
\end{tabular}
\subcaption{}
\label{fig:featurized-results}
\end{subtable}
\begin{subfigure}[c]{0.50\linewidth}
\includegraphics[width=0.9\linewidth]{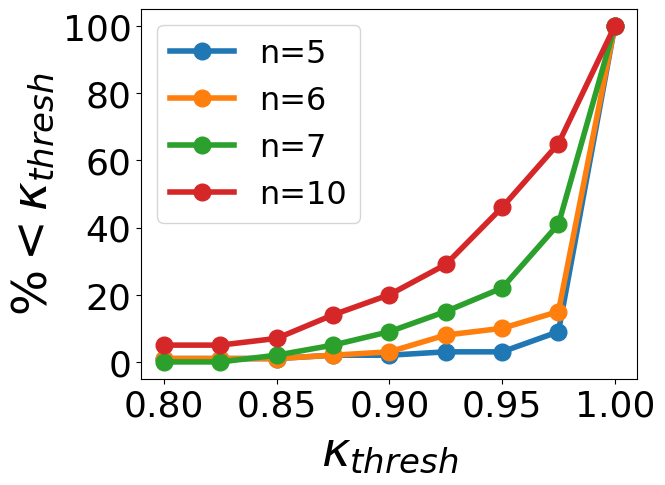}
\subcaption{}
\label{fig:featurized-results-2}
\end{subfigure}
    \label{fig:tantrum_featurized}
\caption{Results for Featurized \textsc{Trantrum} as depth $n$ varies, based on $\kappa$, the ratio of the leader's payoff to the true optimum. (a) grd-$\overline{\kappa}$ and str-$\overline{\kappa}$ denote results for the baseline greedy method and our results when stress tested with $\textbf{q}$ drawn from a distribution from training. (b) Proportion of trials which give $\kappa < \kappa_{\text{thresh}}$. }
\end{figure}
\subsection{Results and Discussion}  
For fixed parameter games whose optimal value can be computed, we observe near optimal performance which significantly outperforms other baselines. For [\textsc{RC}], the average value of each an improvement of .5 is approximately equal to moving an extra half move. In [\textsc{Tantrum}], the subgame perfect equilibrium is vacuous as $\mathsf{P}_1$ is unable to issue threats and gets a payoff of 0. In [\textsc{RC}+], we are unable to fully expand the game tree, however, we still significantly outperform the non-strategic baseline.

For featurized \textsc{Tantrum}, we perform near-optimally for small $n$, even when stress tested with out-of-distribution $\textbf{q}$'s (Figure~\ref{fig:featurized-results}). Performance drops as $n$ becomes larger, which is natural as EPFs become more complex. While performance degrades as $n$ increases, we still significantly outperform the greedy baseline. The stress test suggests that the network is not merely memorizing data. 

Figures~\ref{fig:epoch-comparison-100k} and \ref{fig:epoch-comparison-2M} shows the learned EPFs at the root for epochs 100k and 2M, obtained directly or from one-step lookahead. 
As explained in Section~\ref{sec:our_method}, we only learn the decreasing portions of EPFs. 
After 2M training epochs, the predicted EPFs and one-step lookahead mirrors the true EPF in the decreasing portions, which is not the case at the beginning. 
At the beginning of training, many knots (red markers) are wasted on learning the `useless' increasing portions on the left. After 2M epochs, knots (blue markers) were learning the EPF at the `useful' decreasing regions.

Figure~\ref{fig:threat-large-fa} gives an state in \textsc{Tantrum} whose EPF yields high loss even after training. This failure case is not rare since $\textsc{Tantrum}$ is large. Yet, the resultant action is still optimal---in this case the promise to $\textsf{P}_2$ was $\mu_2=-25.5$ which is precisely $\UBV(s)$. 
Like MDPs, policies can be near-optimal even with high Bellman losses in some states.
\section{Conclusion}
We proposed a novel method of performing FA on EPFs that allows us to efficiently solve for SEFCE. This is to the best of our knowledge, the first time a such an object has been learned from state features, leading to a FA-based method of solving Stackelberg games with performance guarantees.
We hope that our approach will help to close the current gap between solving zero-sum and general-sum games. 

\newpage

\section*{Acknowledgements}
This work was supported by NSF grant IIS-2046640 (CAREER).
\bibliography{bibliography}
\bigskip

\appendix
\appendix
\onecolumn
\section{Proof of Theorem~\ref{thm:FA_guarantee}}
\subsection{Preliminaries}
The proof follows 2 steps. First, we show that the estimate $\tilde{U}_s$ is not too different from the optimal $U_s$. 
Let the depth of a state $D(s)$ be the longest path needed to reach a leaf, i.e.,
\begin{align*}
    D(s) = \begin{cases}
        0 \qquad & s \in \mathcal{L} \\
        \max_{s' \in \mathcal{C}(s)} D(s') + 1 \qquad &\text{otherwise}
    \end{cases}.
\end{align*}
Since $\game$ is finite, $D = \max_{s\in\mathcal{S}}D(s)$ and is finite.

Let $\tilde{U}_s$ be our predicted EPF at state $s$.
For $s \in \mathcal{S} \backslash \mathcal{L}$, we denote (\textit{for this section only}) using shorthand
\begin{align*}
\tilde{U}'_s = \tilde{U}^{\text{target}}_s &= \begin{cases}
    \left[ \bigwedge_{s' \in \mathcal{C}(s)}\tilde{U}_{s'} \right] (\mu) & \text{ if } s \in \mathcal{S}_1 \\
    \left[ \bigwedge_{ s' \in \mathcal{C}(s)} \tilde{U}_{s'} \triangleright \tau(s') \right] (\mu) & \text{ if } s \in \mathcal{S}_2
    \end{cases}
\end{align*}
be what is obtained from one-step lookahead using Section~\ref{sec:review}, and for $s\in\mathcal{L}$,
\begin{align*}
    \tilde{U}_s (\mu_2) = \tilde{U}'_s (\mu_2) = \begin{cases}
        r_1(s) \qquad & \mu_2 = r_2(s) \\
        -\infty \qquad & \text{otherwise}
    \end{cases}.
\end{align*}
For clarity, we denote likewise for the exact EPFs $U'_s$ (which will be equal by definition to $U_s$).

\subsubsection{Domains of EPFs}
Given any real valued function $h: \mathbb{R}\mapsto\mathbb{R}$, we denote its domain by $\dom[h] = \{ x | h (x) > -\infty\}$. 
The following lemmas ensure that the required \textit{domains} all match. This theorem is added for completeness; the reader can skip over this section if desired.
\begin{lemma}
For all $s \in \mathcal{S}$,
$\dom[U_s] = \dom[U'_s] = \dom[\tilde{U}_s] = \dom[\tilde{U}'_s] = [\LBV(s), \UBV(s)]$,
\label{lem:dom_predicted}
\end{lemma}
The proofs are straightforward but trivial, hence they are deferred to Section~\ref{sec:appendix-domains}
\begin{proof}
The first equality was shown in Lemma~\ref{lem:dom_exact}. We can, in fact reuse the proof of Lemma~\ref{lem:dom_exact} for $\tilde{U}'_s$ and $\tilde{U}_s$. This completes this Lemma~\ref{lem:dom_predicted}.
\end{proof}
Hence, we no longer have to worry about mismatched domains. 
\subsubsection{Upper Concave Envelopes and Left Truncations}
First, observe that since $\tilde{U}_s$ is a one-dimensional function, $\bigwedge$ can be written alternatively as
\begin{align*}
    \left[\bigwedge_{s'\in\mathcal{C}(s)} U_{s'}\right] (\mu_2) &= 
    \max_{\substack{s', s'' \in \mathcal{C}(s) \\ t \in [0, 1]; \mu', \mu'' \in \mathbb{R} \\ t\mu' + (1-t)\mu'' = \mu_2}} t U_{s'}(\mu') + (1-t)U_{s''}(\mu''),
\end{align*}
that is, the maximum that could be obtained by interpolating between at most two points across 2 children states $s', s'' \in \mathcal{C}(s)$ (this follows from the fact that all $U$'s are 1-dimensional functions \cite{letchford2010computing}). 

Recall that $\tau(s')$ is defined only for $s' \in \mathcal{C}(s)$, where $s \in \mathcal{S}_2$. For convenience, we define $\beta(s')=\tau(s')$ when $s \in \mathcal{S}_2$ and $-\infty$ when $s \in \mathcal{S}_1$. $\beta$ plays the same role as $\tau$, since left truncating at $-\infty$ does not change anything, i.e., $f \triangleright (-\infty) = f$ for any $f: \mathbb{R} \mapsto \mathbb{R}$. Using $\beta$ allows us to perform a `dummy' left truncation when $s \in \mathcal{S}_1$ and avoid having to split into different cases.

\subsubsection{Our Goal}
Our goal is to show that assuming that the $L_\infty$ loss is low for all states $s$, i.e., 
\begin{align}
\max_{\mu_2 \in [\LBV(s), \UBV(s)]}|\tilde{U}_s(\mu_2)- \tilde{U}'_s(\mu_2)| \leq \epsilon
\label{eq:proof_assumption}
\end{align}
then we will enjoy good performance, i.e, the leader gets a payoff of order $\mathcal{O}(\epsilon D)$ less than optimal (additively).

\subsection{Learned EPFs are Approximately Optimal}
The first half of the theorem is to show that our learned EPFs $\tilde{U}_s$ are for all $s$, close to the true $U_s$ pointwise.
We prove the main theorem by strong induction on the states by increasing depth the following.
and \eqref{eq:sefce_recurse}. Our induction hypothesis is
\begin{align*}
    \mathcal{K}_j: &\quad \max_{\mu_2 \in [\LBV(s), \UBV(s)]} | \tilde{U}_s(\mu_2) - U_s (\mu_2) | \leq j \epsilon \qquad \forall s \text{ where } D(s) = j.
\end{align*}
By definition, $K_0$ satisfies our requirement since $s \in \mathcal{L}$. Thus, the base case is satisfied. Now we prove the inductive case. Assume that $\mathcal{K}_0, \cdots, \mathcal{K}_{j-1}$ are all satisfied. We want to show $\mathcal{K}_j$ using \eqref{eq:proof_assumption}.
\begin{lemma}
Let $s \in \mathcal{S}$ such that $D(s) = j$. Suppose $\mathcal{K}_{0}, \dots \mathcal{K}_{j-1}$ are true. Then we have $| \tilde{U}'_s(\mu_2) - U'_s (\mu_2) | \leq \epsilon(j-1)$ for all $\mu_2 \in [\LBV(s), \UBV(s)]$.
\label{lem:inductive}
\end{lemma}
\begin{proof}

Fix $\mu_2$. We want to show $| \tilde{U}'_s(\mu_2) - U'_s (\mu_2) | \leq \epsilon(j-1)$. 
Now, let $\hat{\sigma} = (\hat{s}',\hat{s}'',\hat{t}, \hat{\mu}', \hat{\mu}'')$ be the parameters which achieves the maximum
\begin{align}
\argmax_{\substack{s', s'' \in \mathcal{C}(s) \\ t \in [0, 1]; \mu', \mu'' \in \mathbb{R} \\ t\mu' + (1-t)\mu'' = \mu_2}} t [U_{s'}\triangleright \beta(s')](\mu') + (1-t)[U_{s''} \triangleright \beta(s'')] (\mu'')
\label{eq:proof_true_convex_argmax}
\end{align}
and similarly when we are working with learned EFPs, $\tilde{\sigma} = (\tilde{s}',\tilde{s}'',\tilde{t}, \tilde{\mu}', \tilde{\mu}'')$
\begin{align}
\argmax_{\substack{s', s'' \in \mathcal{C}(s) \\ t \in [0, 1]; \mu', \mu'' \in \mathbb{R} \\ t\mu' + (1-t)\mu'' = \mu_2}} 
t [\tilde{U}_{s'} \triangleright \beta(s')](\mu') + (1-t)[\tilde{U}_{s''} \triangleright \beta(s'')](\mu''),
\label{eq:proof_estimate_convex_argmax}
\end{align}
which are the arguments that optimized give
$U'_s(\mu_2)$ and
$\tilde{U}'_s(\mu_2)$
respectively. That is, $\hat{\sigma}$ and $\tilde{\sigma}$ gives how the point at the EPF with the x coordinate equal to $\mu_2$ is obtained as a mixture of at most 2 points from the upper convex envelope  (when $s'=s''$, we simply repeat the 2 points and set $t=1/2$ for simplicity). We proceed by showing a contradiction. We have two cases. 

\paragraph{Case 1.} Suppose that 
$\tilde{U}'_s(\mu_2) > U'_s(\mu_2) + \epsilon(j-1)$. Then we have
\begin{align}
    \begin{split}
        & \quad \left|
        \underbrace{\tilde{t} 
        [\tilde{U}_{\tilde{s}'} \triangleright \beta(\tilde{s}')](\tilde{\mu}') + (1-\tilde{t})[\tilde{U}_{\tilde{s}''} \triangleright \beta(\tilde{s}'')] (\tilde{\mu}'')}_{=\tilde{U}'_{s}(\mu_2) > U'_s(\mu_2)+\epsilon(j-1)} - 
        \underbrace{\tilde{t} 
        [U_{\tilde{s}'}\triangleright \beta(\tilde{s}')](\tilde{\mu}') + (1-\tilde{t})
        [U_{\tilde{s}''}\triangleright \beta(\tilde{s}'')]( \tilde{\mu}'')}_{\leq U'_s(\mu_2)}
        \right| \\ 
        &= 
        \left|
        \tilde{t} 
        \left( [\tilde{U}_{\tilde{s}'} \triangleright \beta(\tilde{s}')](\tilde{\mu}') 
        -[U_{\tilde{s}'}\triangleright \beta(\tilde{s}')](\tilde{\mu}')
        \right) + 
        (1-\tilde{t}) 
        \left( [\tilde{U}_{\tilde{s}''}\triangleright \beta(\tilde{s}'')]( \tilde{\mu}'') - 
        [U_{\tilde{s}''}\triangleright \beta(\tilde{s}'')](\tilde{\mu}'') 
        \right)
        \right| \\ 
        &\leq \epsilon (j-1) 
    \end{split}
\end{align}
and where first inequality inside $|\cdot|$ follows from our assumption in case 1, and the second from the fact that $U'_s(\mu_2)$ was taken from an argmax, i.e., \eqref{eq:proof_true_convex_argmax}. The third line holds from our induction hypothesis $\mathcal{K}_i$ \eqref{eq:proof_assumption}, where $i \in [0, j-1]$, the fact that $D(s') < D(s)=j$ and how the $\triangleright$ operator cannot increase the absolute error of the difference.\footnote{Note that since $\tilde{\sigma}$ is the argmax, we are guaranteed to not have values of $\tilde{\mu}$ that lie outside the domains of the truncated EPFs.} However, these $3$ inequalities cannot hold simultaneously, thus the assumption for Case 1 cannot be true, i.e., we must have $\tilde{U}'_s(\mu_2) \leq U'_s(\mu_2) + \epsilon(j-1)$

\paragraph{Case 2.} Suppose that $\tilde{U}'_s(\mu_2) < U'_s(\mu_2) - \epsilon(j-1)$. Then using a similar derivation we have
\begin{align}
    \begin{split}
        & \quad \left|
        \underbrace{\hat{t} [U_{\hat{s}'} \triangleright \beta(\hat{s}')](\hat{\mu}') + (1-\hat{t})[]U_{\hat{s}''}\triangleright \beta(\hat{s}'')]( \hat{\mu}'')}_{= U'_{s}(\mu_2) > \tilde{U}'_s(\mu_2)+\epsilon(j-1)} - 
        \underbrace{\hat{t} [\tilde{U}_{\hat{s}'}\triangleright \beta(\hat{s}')](\hat{\mu}') + (1-\hat{t})[\tilde{U}_{\hat{s}''} \triangleright \beta(\hat{s}'')](\hat{\mu}'')}_{\leq \tilde{U}'_s(\mu_2)}
        \right| \\
        &= \left| 
        \hat{t} \left( [U_{\hat{s}'}\triangleright \beta(\hat{s}')](\hat{\mu}') - [\tilde{U}_{\hat{s}'} \triangleright \beta(\hat{s}')](\hat{\mu}') \right) 
        + (1-\hat{t}) \left( [U_{\hat{s}''}\triangleright \beta(\hat{s}'')]( \hat{\mu}'') 
        -[\tilde{U}_{\hat{s}''}\triangleright \beta(\hat{s}'')](\hat{\mu}'')
        \right)
        \right| \\
        &\leq \epsilon (j-1) 
    \end{split}
\end{align}
where the first inequality inside $|\cdot|$ comes from case 2's assumption, the second is from the fact that $\tilde{U}_s'(\mu_2)$ was taken from an argmax, i.e., \eqref{eq:proof_estimate_convex_argmax}. The third line follows from the induction hypothesis $\mathcal{K}_i$, where $i\in [0, j-1]$ and the depth of $s$. These 3 inequalities cannot hold simultaneously, so our assumption in case 2 cannot be true and $\tilde{U}'_s(\mu_2) \geq U'_s(\mu_2) - \epsilon(j-1)$. 

Combining the result from both cases gives the desired result. 
\end{proof}

Now, we consider the $\mu_2$ which has the worst possible discrepancy, which gives
\begin{align}
\begin{split}
    &\max_{\mu_2 \in [\LBV(s), \UBV(s)]} | \tilde{U}_s(\mu_2) - U_s (\mu_2) | \\
    \leq& \max_{\mu_2 \in [\LBV(s), \UBV(s)]} 
    | \tilde{U}_s(\mu_2) - \tilde{U}'_s (\mu_2) | + 
    | \tilde{U}'_s(\mu_2) - U_s (\mu_2) | \\
    =& \max_{\mu_2 \in [\LBV(s), \UBV(s)]} 
    | \tilde{U}_s(\mu_2) - U'_s (\mu_2) | + 
    | \tilde{U}'_s(\mu_2) - U'_s (\mu_2) | \\
    \leq& \max_{\mu_2 \in [\LBV(s), \UBV(s)]} 
    | \tilde{U}_s(\mu_2) - \tilde{U}'_s (\mu_2) | + 
    \max_{\mu_2 \in [\LBV(s), \UBV(s)]}
    | \tilde{U}'_s(\mu_2) - U'_s (\mu_2) | \\
    \leq& 
    \max_{\mu_2 \in [\LBV(s), \UBV(s)]}
    | \tilde{U}'_s(\mu_2) - U'_s (\mu_2) | + \epsilon \\
    \leq& \epsilon j,
\end{split}
\label{eq:main_proof_midway}
\end{align}
where the second line follows from the triangle inequality, the third line using the equality between $U_s(\mu_2)$ and $U'_s(\mu_2)$, the fourth from the fact that $\max_x |f(x) + g(x)| \leq \max_x|f(x)| + \max_y |g(y)|$, the fifth from the assumption \eqref{eq:proof_assumption}, the the last line from Lemma~\ref{lem:inductive}. 
The main theorem follows by induction on $D(s)$, the fact that $D(s)$ is bounded by the depth of the tree, and the fact that the base case $\mathcal{K}_0$ is trivially true.
\begin{theorem}
If $L_\infty (\tilde{U}_s, \tilde{U}'_s) \leq \epsilon$ for all $s \in \mathcal{S}$, then $\max_{\mu_2 \in [\LBV(s), \UBV(s)]} | \tilde{U}_s(\mu_2) - U_s (\mu_2) | \leq \epsilon D$, where $D$ is the depth of the game.
\label{thm:ours_almost_accurate}
\end{theorem}
\begin{proof}
This follows by the definition of $L_\infty$ and the above derivations.
\end{proof}
\subsection{Leader's Payoffs from Induced Strategy is Close-To-Expected}
\label{sec:appendix_induced_gives_close_payoffs}
Theorem~\ref{thm:ours_almost_accurate} tells us that our EPF everywhere is close (pointwise) to the true EPF if $\epsilon$ is small. Now we need to establish the suboptimality when \textit{playing} according to $\tilde{\pi}$, which is a joint policy implicit from $\tilde{U}_s$. We begin with some notation. 

Let $\tilde{Q}_s(\mu_2) : [\LBV(s), \UBV(s)] \mapsto \mathbb{R}$ be the payoff to $\mathsf{P}_1$ assuming we started at state $s$, promised a payoff of $\mu_2$ to $\mathsf{P}_2$ and used the approximate EPFs $\tilde{U}_s$ for all descendent states $s' \sqsupseteq s$ (the domain precludes unfulfilled promises). That is, given $\tilde{\sigma} = (\tilde{s}',\tilde{s}'',\tilde{t}, \tilde{\mu}', \tilde{\mu}'')$ given by
\begin{align}
\argmax_{\substack{s', s'' \in \mathcal{C}(s) \\ t \in [0, 1]; \mu', \mu'' \in \mathbb{R} \\ t\mu' + (1-t)\mu'' = \mu_2}} t [\tilde{U}_{s'}\triangleright \beta(s')](\mu') + (1-t)[\tilde{U}_{s''}\triangleright \beta(s'')](\mu''),
\end{align}
the induced policy is to play to $\tilde{s}'$ with probability $\tilde{t}$ and a consequent promise of $\tilde{\mu}'$, as well as playing to  $\tilde{s}''$ with probability $(1-\tilde{t})$ and a promised payoff of $\tilde{\mu}''$. By definition, if $s \in \mathcal{L}$, $\tilde{Q}_s = \tilde{U}_s = U_s$ trivially. We also have the following recursive equations for a given $s \not \in \mathcal{L}$, $\mu_2$
\begin{align}
    \tilde{Q}_s(\mu_2) &= \tilde{t} \tilde{Q}_{\tilde{s}'}(\tilde{\mu}') + (1-\tilde{t})\tilde{Q}_{\tilde{s}''}(\tilde{\mu}'').
    \label{eq:q_recurse}
\end{align}

\begin{theorem}
$\left| \tilde{Q}_s(\mu_2) - \tilde{U}_s(\mu_2) \right| \leq \epsilon D$ for all $s \in \mathcal{S}$ and for all $\mu_2 \in [\LBV(s), \UBV(s)]$.
\label{thm:Q_diff_U}
\end{theorem}
\begin{proof}
The proof is given by strong induction on the $D(s)$ again. Let 
\begin{align}
    \mathcal{H}_j : \quad \left| \tilde{Q}_s(\mu_2) - \tilde{U}_s(\mu_2) \right| \leq \epsilon j \quad \forall s \text{ where } D(s) = j
\end{align}
By definition $\mathcal{H}_0$ is true. Now let us suppose that $\mathcal{H}_0 \dots \mathcal{H}_{j-1}$ is true. We have for $s \in \mathcal{S}, D(s) = j$,
\begin{align*}
    &  \left| \tilde{Q}_s(\mu_2) - \tilde{U}_s(\mu_2) \right| \\
    \leq& \left| \tilde{Q}_s(\mu_2) - \tilde{U}'_s(\mu_2) \right| + 
    \left| \tilde{U}'_s(\mu_2) - \tilde{U}_s(\mu_2) \right| \\
    \leq& \left| \tilde{Q}_s(\mu_2) - \tilde{U}'_s(\mu_2) \right| + 
    \epsilon \\
    =& \left| \tilde{t} \tilde{Q}_{\tilde{s}'}(\tilde{\mu}') + (1-\tilde{t})\tilde{Q}_{\tilde{s}''}(\tilde{\mu}'') - 
    \tilde{t} [\tilde{U}_{\tilde{s}'}\triangleright \beta(\tilde{s}')](\tilde{\mu}') - (1-\tilde{t})[\tilde{U}_{\tilde{s}''} \triangleright\beta(\tilde{s}'')](\tilde{\mu}'')
    \right| + \epsilon\\
    \leq& 
    \tilde{t}
    \underbrace{\left| 
     \tilde{Q}_{\tilde{s}'}(\tilde{\mu}') - \tilde{U}_{\tilde{s}'}(\tilde{\mu}') 
    \right|}_{\leq \epsilon (j-1)} + 
    (1-\tilde{t})
    \underbrace{\left|
    \tilde{Q}_{\tilde{s}''}(\tilde{\mu}'') - \tilde{U}_{\tilde{s}''}(\tilde{\mu}'')
    \right|}_{\leq \epsilon (j-1)} + \epsilon\\
    \leq& \epsilon j
\end{align*}
The second line follows from the triangle inequality. The third line follows from our FA assumption \eqref{eq:proof_assumption}. The fourth line follows from expansion of the definitions of $\tilde{Q}_s$ and $\tilde{U}'_s$, i.e., \eqref{eq:q_recurse} and \eqref{eq:proof_estimate_convex_argmax}
 The fifth line follows the induction hypothesis and the fact that $s', s'' \in \mathcal{C}(s)$ have at least one lower depth than $s$. Also, the truncation operator never causes any element to exceed domain bounds (which would give $-\infty$ values). By strong induction $\mathcal{H}_j$ is true for all $j \in [0, D]$ and the theorem follows through directly.
\end{proof}
\subsection{Piecing Everything Together}
Let $\mu_2^* = \argmax_{\mu_2} U_{\text{root}}(\mu_2)$, i.e., the promise given to the follower \textit{at the root} under the optimal policy $\pi^*$. Let $\tilde{\mu}_2 = \argmax_{\mu_2} \tilde{U}_\text{root} (\mu_2) $, which is the promise to be given to the follower \textit{at the root} under $\tilde{\pi}$. We have
\begin{align*}
    & U_\text{root} (\mu_2^*) - \tilde{Q}_\text{root}(\tilde{\mu}_2) \\
    =& \underbrace{U_\text{root} (\mu_2^*) - \tilde{U}_\text{root} (\mu_2^*)}_{|\cdot| \leq \epsilon D}
    + \underbrace{\tilde{U}_\text{root} (\mu_2^*)}_{\leq \tilde{U}_\text{root}(\tilde{\mu}_2)} - \tilde{Q}_\text{root}(\tilde{\mu}_2) \\
    \leq& \epsilon D + \left| \tilde{U}_\text{root}(\tilde{\mu}_2) - \tilde{Q}_\text{root}(\tilde{\mu}_2) \right|\\
    \leq& 2\epsilon D.
\end{align*} 
The inequalities in the second line come from Theorem~\ref{thm:ours_almost_accurate} and the definition of $\tilde{\mu}_2$; specifically that it is taken over the argmax. The last line comes from Theorem~\ref{thm:Q_diff_U}. To complete the proof, we simply observe that $\tilde{Q}_\text{root}(\tilde{\mu}_2)$ is precisely $r_1(\tilde{\pi})$ by definition.

\section{Proof of Theorem~\ref{thm:FA_guarantee_approx}}
The proof is essentially the same as presented in Section~\ref{sec:appendix_induced_gives_close_payoffs}, except that we are working with approximate bounds rather than strict ones. We have, using the new bounds,

For $s \in \mathcal{S} \backslash \mathcal{L}$, we denote using shorthand
\begin{align*}
\tilde{U}'_s = \tilde{U}^{\text{target}}_s &= \begin{cases}
    \left[ \bigwedge_{s' \in \mathcal{C}(s)}\tilde{U}_{s'} \right] (\mu) & \text{ if } s \in \mathcal{S}_1 \\
    \left[ \bigwedge_{ s' \in \mathcal{C}(s)} \tilde{U}_{s'} \triangleright \tilde{\tau}(s') \right] (\mu) & \text{ if } s \in \mathcal{S}_2
    \end{cases}
\end{align*}
be what is obtained from one-step lookahead using Section~\ref{sec:review}, and for $s\in\mathcal{L}$,
\begin{align*}
    \tilde{U}_s (\mu_2) = \tilde{U}'_s (\mu_2) = \begin{cases}
        r_1(s) \qquad & \mu_2 = r_2(s) \\
        -\infty \qquad & \text{otherwise}
    \end{cases}.
\end{align*}
This is the same as the case with exact $\LBV_s, \UBV_s$, but with a stricter truncation $\tilde{\tau}(s')$ for each $s' \in \mathcal{C}(s)$. We define $\tilde{\beta}$ just like before: $\tilde{\beta}(s')=\tilde{\tau}(s')$ when $s \in \mathcal{S}_2$ and $-\infty$ when $s \in \mathcal{S}_1$. 
We follow along the same way as Theorem~\ref{thm:FA_guarantee} in Section~\ref{sec:appendix_induced_gives_close_payoffs}.

Let $\tilde{Q}_s(\mu_2) : [\ALBV(s), \AUBV(s)] \mapsto \mathbb{R}$ be the payoff to $\mathsf{P}_1$ assuming we started at state $s$, promised a payoff of $\mu_2$ to $\mathsf{P}_2$ and used the approximate EPFs $\tilde{U}_s$ for all descendent states $s' \sqsupseteq s$ (the domain precludes unfulfilled promises). That is, given $\tilde{\sigma} = (\tilde{s}',\tilde{s}'',\tilde{t}, \tilde{\mu}', \tilde{\mu}'')$ given by
\begin{align}
\argmax_{\substack{s', s'' \in \mathcal{C}(s) \\ t \in [0, 1]; \mu', \mu'' \in \mathbb{R} \\ t\mu' + (1-t)\mu'' = \mu_2}} t [\tilde{U}_{s'}\triangleright \tilde{\beta}(s')](\mu') + (1-t)[\tilde{U}_{s''}\triangleright \tilde{\beta}(s'')](\mu''),
\end{align}
the induced policy is to play to $\tilde{s}'$ with probability $\tilde{t}$ and a consequent promise of $\tilde{\mu}'$, as well as playing to  $\tilde{s}''$ with probability $(1-\tilde{t})$ and a promised payoff of $\tilde{\mu}''$. By definition, if $s \in \mathcal{L}$, $\tilde{Q}_s = \tilde{U}_s = U_s$ trivially. Just like before, we also have the following recursive equations for a given $s \not \in \mathcal{L}$, $\mu_2$
\begin{align}
    \tilde{Q}_s(\mu_2) &= \tilde{t} \tilde{Q}_{\tilde{s}'}(\tilde{\mu}') + (1-\tilde{t})\tilde{Q}_{\tilde{s}''}(\tilde{\mu}'').
    \label{eq:q_recurse_approx}
\end{align}

Let our induction hypothesis be
\begin{align}
    \mathcal{H}_j : \quad \left| \tilde{Q}_s(\mu_2) - \tilde{U}_s(\mu_2) \right| \leq \epsilon j \quad \forall s \text{ where } D(s) = j.
\end{align}
By definition $\mathcal{H}_0$ is true. Now let us suppose that $\mathcal{H}_0 \dots \mathcal{H}_{j-1}$ is true. We have for $s \in \mathcal{S}, D(s) = j$,
\begin{align*}
    &  \left| \tilde{Q}_s(\mu_2) - \tilde{U}_s(\mu_2) \right| \\
    \leq& \left| \tilde{Q}_s(\mu_2) - \tilde{U}'_s(\mu_2) \right| + 
    \left| \tilde{U}'_s(\mu_2) - \tilde{U}_s(\mu_2) \right| \\
    \leq& \left| \tilde{Q}_s(\mu_2) - \tilde{U}'_s(\mu_2) \right| + 
    \epsilon \\
    =& \left| \tilde{t} \tilde{Q}_{\tilde{s}'}(\tilde{\mu}') + (1-\tilde{t})\tilde{Q}_{\tilde{s}''}(\tilde{\mu}'') - 
    \tilde{t} [\tilde{U}_{\tilde{s}'}\triangleright \tilde{\beta}(\tilde{s}')](\tilde{\mu}') - (1-\tilde{t})[\tilde{U}_{\tilde{s}''} \triangleright \tilde{\beta}(\tilde{s}'')](\tilde{\mu}'')
    \right| + \epsilon\\
    \leq& 
    \tilde{t}
    \underbrace{\left| 
     \tilde{Q}_{\tilde{s}'}(\tilde{\mu}') - \tilde{U}_{\tilde{s}'}(\tilde{\mu}') 
    \right|}_{\leq \epsilon (j-1)} + 
    (1-\tilde{t})
    \underbrace{\left|
    \tilde{Q}_{\tilde{s}''}(\tilde{\mu}'') - \tilde{U}_{\tilde{s}''}(\tilde{\mu}'')
    \right|}_{\leq \epsilon (j-1)} + \epsilon\\
    \leq& \epsilon j
\end{align*}
The second line follows from the triangle inequality. The third line follows from our FA assumption \eqref{eq:proof_assumption}. The fourth line follows from expansion of the definitions of $\tilde{Q}_s$ and $\tilde{U}'_s$, i.e., \eqref{eq:q_recurse} and \eqref{eq:proof_estimate_convex_argmax}
 The fifth line follows the induction hypothesis and the fact that $s', s'' \in \mathcal{C}(s)$ have at least one lower depth than $s$. Also, the truncation operator never causes any element to exceed domain bounds (which would give $-\infty$ values). By strong induction $\mathcal{H}_j$ is true for all $j \in [0, D]$. Finally, we observe that $\tilde{Q}_s(\mu_2) = R_1(\tilde{\pi})$ when $\tilde{\mu}_2 = \argmax_{\mu_2} \tilde{U}_s(\mu_2)$. This completes the proof.

\section{Proof of Lemma~\ref{lem:dom_predicted}}
\label{sec:appendix-domains}
\begin{lemma}
For all $s \in \mathcal{S}$, $\dom[U_s]$ = $\dom[U's]=[\LBV(s), \UBV(s)]$.  
\label{lem:dom_exact}
\end{lemma}
\begin{proof}
The first equality is by definition. We now show that 
$\dom[U_s] = [\LBV, \UBV]$ by definition. Consider the state $s$ we are applying \eqref{eq:sefce_recurse} to and the 2 possible cases.

\paragraph{Case 1: $s \in \mathcal{S}_1$.} By definition $\UBV(s) = \max_{s' \in \mathcal{C}(s)}\UBV(s')$, and $\LBV(s) = \min_{s' \in \mathcal{C}(s)} \LBV(s)$. First, observe that 
\begin{align}
\begin{split}
\max \left\{ \dom \left[ U'_s \right] \right\} 
&= \max \left\{ \dom \left[ \bigwedge_{s' \in \mathcal{C}(s)} U_{s'} \right]\right\} \\ 
&= \max_{s' \in \mathcal{C}(s)} \max \left\{ \dom \left[ U_{s'} \right] \right\} \\ 
&= \max_{s'\in\mathcal{C}(s)} \UBV(s') \\ 
&= \UBV(s),
\end{split}
\label{eq:dom_upper_case_1}
\end{align}
where the second line follows from the fact that the largest x-coordinate after taking the upper-concave-envelope is the largest of the largest-x coordinates over each $U_{s'}$. Similarly, we have
\begin{align}
\begin{split}
\min \left\{ \dom \left[ U'_s \right] \right\} 
&= \min \left\{ \dom \left[ \bigwedge_{s' \in \mathcal{C}(s)} U_{s'} \right]\right\} \\ 
&= \min_{s' \in \mathcal{C}(s)} \min \left\{ \dom \left[ U_{s'} \right] \right\} \\ 
&= \min_{s'\in\mathcal{C}(s)} \LBV(s') \\ 
&= \LBV(s)
,
\end{split}
\label{eq:dom_lower_case_1}
\end{align}
where the second line comes again from the fact that the lowest x-coordinate after taking upper concave envelopes is the smallest of all the smallest x-coordinates over each $U_{s'}$. Now, \eqref{eq:dom_upper_case_1} and \eqref{eq:dom_lower_case_1} established the lower and upper limits of $U'_s$. 
Since $U'_s$ is concave, for every $\min \left\{ \dom \left[ U'_s \right] \right\} \leq \mu_2 \leq \max \left\{ \dom \left[ U'_s \right] \right\} $ we have $U'_s (\mu_2) \geq \min \left( U'_s(\min \left\{ \dom \left[ U'_s \right] \right\}), U'_s(\max \left\{ \dom \left[ U'_s \right] \right\}) \right) > -\infty$. This completes Case 1.

\paragraph{Case 2: $s \in \mathcal{S}_2$.} By definition, $\UBV(s)=\max_{s' \in \mathcal{C}(s)}\UBV(s')$ and $\LBV(s) = \max_{s' \in \mathcal{C}(s)} \LBV(s')$ (note the difference with Case 1, since $\mathsf{P}_2$ decides the current action). We again work out upper and lower bounds of $\dom [U'_s]$. 
\begin{align}
\begin{split}
    \max \left\{ \dom \left[ U'_s \right] \right\} 
    &= \max \left\{ \dom \left[ \bigwedge_{s' \in \mathcal{C}(s)} U_{s'} \triangleright \tau(s') \right]\right\} \\ 
    &= \max_{s' \in \mathcal{C}(s)} \max \left\{ \dom \left[ U_{s'} \triangleright \tau(s') \right] \right\} \\ 
    &= \max_{s' \in \mathcal{C}(s)} \max  \left\{ [\LBV(s'), \UBV(s')] \cap [\tau(s'), \infty) \right\} \\ 
    &= \max_{s'\in\mathcal{C}(s)} \UBV(s') \\ 
    &= \UBV(s),
\end{split}
\label{eq:dom_upper_case_2}
\end{align}
where the fourth line follows from the fact that $= \max_{s'\in\mathcal{C}(s)} \UBV(s') \geq \max_{s^! \in \mathcal{C}(s); s^! \neq s'} \UBV(s^!) \geq 
\max_{s^! \in \mathcal{C}(s); s^! \neq s'} \LBV(s^!) = \tau(s')$ (i.e., that the highest x coordinate in $U'_s$ is never part of the left-truncation step). For any $s'\in\mathcal{C}(s)$,
\begin{align}
    \dom \left[ U_{s'} \triangleright \tau(s') \right]
    &= 
    \begin{cases}
        \emptyset \qquad & \max \left\{ \dom [U_{s'}] \right\} < \tau(s') \\
        \dom \left[ U_{s'} \right] \qquad & \tau(s') < \min \left\{ \dom \left[ U_{s'} \right] \right\} \\
        [\tau(s'), \max \left\{ \dom [U_{s'}] \right\}] \qquad & \text{otherwise}
    \end{cases}.
    \label{eq:dom_case_2_trunc_cases}
\end{align}
For $s'$ where $\dom \left[ U_{s'} \triangleright \tau(s') \right] \neq \emptyset$, we have
\begin{align}
    \min \left\{ \dom \left[ U_{s'} \triangleright \tau(s') \right] \right\} 
    &= \max_{s'' \in \mathcal{C}(s)} \min \left\{ \dom \left[ U_{s''}\right] \right\}.
    \label{eq:dom_case_2_const}
\end{align}
Note that this is not dependent on $s'$. Also, note that it cannot be the case that $\dom \left[ U_{s'} \triangleright \tau(s') \right] = \emptyset$ for \textit{all} $s'$. In particular, consider $s^* = \argmax_{s' \in \mathcal{C}(s)} \max \left\{\dom [U_{s'}] \right\}$, clearly, $\max \left\{\dom [U_{s'}] \right\} \geq \tau(s^*)$ so we do not wind up with the empty set in \eqref{eq:dom_case_2_trunc_cases}. For simplicity, let $\min  \emptyset = \infty$. Hence, we can write
\begin{align}
\begin{split}
    \min \left\{ \dom \left[ U'_s \triangleright \tau(s') \right] \right\} 
    &= \min \left\{ \dom \left[ \bigwedge_{s' \in \mathcal{C}(s)} U_{s'} \triangleright \tau(s')\right]\right\} \\ 
    &= \min_{s' \in \mathcal{C}(s); } \min \left\{ \dom \left[ U_{s'} \triangleright \tau(s') \right] \right\} \\ 
    &= \min_{s'\in\mathcal{C}(s)} \max_{s'' \in \mathcal{C}(s)} \min \left\{ \dom \left[ U_{s''}\right] \right\} \\ 
    &= \max_{s'' \in \mathcal{C}(s)} \min \left\{ \dom \left[ U_{s''}\right] \right\} \\ 
    &= \LBV(s).
\end{split}
\label{eq:dom_lower_case_2}
\end{align}
The first line is by definition. The second line uses the same argument as in case 1. The second line follows \eqref{eq:dom_case_2_const} and the fact that at least one $s^*$ exists. The last line follows from the definition of $\LBV(s)$. As with case 1, we use \eqref{eq:dom_upper_case_2}, \eqref{eq:dom_lower_case_2} and the fact that $U'_s$ is concave to show that $U'_{s}(\mu_2) > -\infty$ for $\mu_2 \in [\LBV(s), \UBV(s)]$. This completes the proof.
\end{proof}
We are now ready to tackle the proof of Lemma~\ref{lem:dom_predicted} (reproduced here):
For all $s \in \mathcal{S}$,
$$\dom[U_s] = \dom[U'_s] = \dom[\tilde{U}_s] = \dom[\tilde{U}'_s] = [\LBV(s), \UBV(s)],$$
\begin{proof}
The first equality was shown in Lemma~\ref{lem:dom_exact}. We can, in fact reuse the proof of Lemma~\ref{lem:dom_exact} by replacing $U'_s$ and $U_s$ with $\tilde{U}'_s$ and $\tilde{U}_s$. This completes this Lemma~\ref{lem:dom_predicted}.
\end{proof}

\section{Extensions to Games with Chance}
For games with chance, backups will involve infimal convolutions \cite{Cermak:2016:UCS:3015812.3015879}. Denote the set of chance nodes by $\mathcal{S}_{\mathsf{C}}$. For $s\in \mathcal{S}_{\mathsf{C}}$ and we denote the probability of transition for from $s$ to $s'$ to be $\pi_{\textsf{C}} (s', s)$. The set of equations at \eqref{eq:sefce_recurse} has to be augmented by the case where $s\in \mathcal{S}_{\mathsf{C}}$. In these cases, we have
\begin{align*}
    U_s(\mu) = \bigoplus_{s' \in \mathcal{C}(s)} \pi_{\textsf{C}}(s', s) U_{s'}(\mu/\pi_{\textsf{C}}(s', s)), 
\end{align*}
where $\bigoplus$ is the maximal-convolution operator (similar to the inf-conv operator for convex functions),
\begin{align*}
    f_1 \bigoplus f_2 (\mu) &= \sup_y \left\{ f_1(\mu-y) + f_2(y) | y\in \mathbb{R} \right\}.
\end{align*}
Refer to \cite{Cermak:2016:UCS:3015812.3015879} for more details as to why $\bigoplus$ is the right operator to be used. It is well known that $\bigoplus$ can be efficiently implemented (linear in the number of knots) when functions are piecewise linear concave (sort the line segments in all $f_i$ based on gradients and stitch these line segments together in ascending order of gradients). Thankfully, this holds for EPFs. Furthermore, applying $\bigoplus$ to piecewise linear concave functions gives another piecewise linear concave function. Hence, EPFs of $\tilde{U}_s$ for each state and trained again using the $L_\infty$ loss. Theorem~\ref{thm:FA_guarantee} still holds with some minor additions (we omit the proof in this paper).

\section{Comparison between EPFs between SSE and SEFCE}
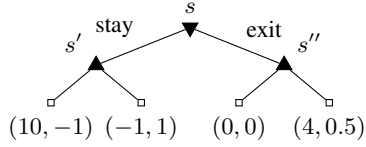
\begin{figure}
    \centering
    \centering
    \begin{tikzpicture}[scale=1,font=\footnotesize]
\tikzstyle{leader}=[regular polygon,regular polygon sides=3,draw,inner sep=1.2,fill=black];
\tikzstyle{follower}=[regular polygon,regular polygon sides=3, rotate=180,draw,inner sep=1.2, fill=black];
\tikzstyle{terminal}=[draw,inner sep=1.2];
\tikzstyle{level 1}=[level distance=5mm,sibling distance=25mm]
\tikzstyle{level 2}=[level distance=5mm,sibling distance=12mm]

\node(0)[follower,label=below:{$s$}]{}
    child{node(1)[leader, label=above left:$s'$]{}
        child{node(3)[terminal, label=below:{$(10,-1)$}]{}
        }
        child{node(4)[terminal, label=below:{$(-1,1)$}]{}
        }
        edge from parent node[above left]{stay}
    }
    child{node(2)[leader, label=above right:$s''$]{}
        child{node(5)[terminal, label=below:{$(0,0)$}]{}
        }
        child{node(6)[terminal, label=below:{$(4,0.5)$}]{}
        }
        edge from parent node[above right]{exit}
    };
\end{tikzpicture}
    \phantomcaption{}
    \label{fig:sefce_eg}
    \caption{Sample game for difference between SSE and SEFCE.}
\end{figure}
\begin{figure}[t]
    \begin{subfigure}{.3\textwidth}
    \centering
    \includegraphics[width=\textwidth]{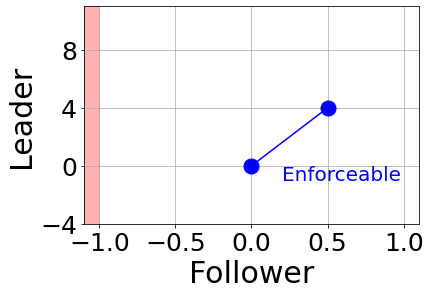}
    \phantomsubcaption{}
    \label{fig:demo2-2}
    \end{subfigure}
    \begin{subfigure}{.3\textwidth}
    \centering
    \includegraphics[width=\textwidth]{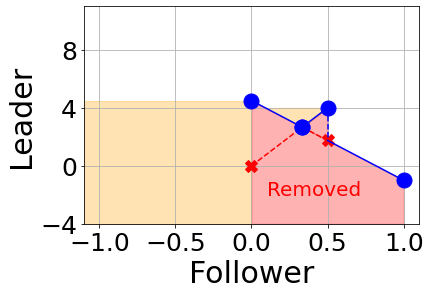}
    \phantomsubcaption{}
    \label{fig:demo2-3}
    \end{subfigure}
    \begin{subfigure}{.3\textwidth}
    \centering
    \includegraphics[width=\textwidth]{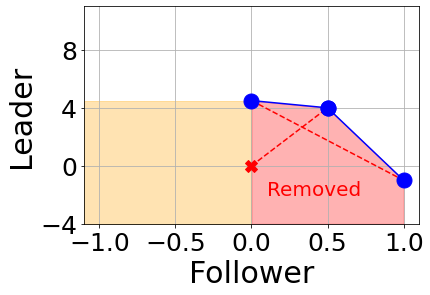}
    \phantomsubcaption{}
    \label{fig:demo2-4}
    \end{subfigure}
    \caption{From left to right: EPFs based on the game in Figure~\ref{fig:sefce_eg}. (a) Enforceable EPF at $s''$, (b) EPF for SSE at $s$, (c) EPF for SEFCE at $s$. The EPF at $s'$ is the same as that in Figure~\ref{fig:demo-2}.}
    \label{fig:demo2}
\end{figure}
One of the key disadvantages of EPFs in SSE is that they could be non-concave, or worst still, discontinuous. Consider the game in Figure~\ref{fig:sefce_eg} and the EPFs in Figure~\ref{fig:demo2}. 
In Figure~\ref{fig:demo2-3}, we can see that the EPF is neither concave or even continuous.
This is because the SSE takes pointwise maximums at follower nodes and not upper-concave envelopes. On the other hand, Figure~\ref{fig:demo2-4} shows the EPF of an SEFCE. Here, it is much better behaved, being a piecewise linear concave function.Furthermore, as mentioned in Theorem~\ref{thm:plc}, the EPF in SEFCE dominates (is always higher or equal to at all x-coordinates) the EPF of SSE. This means the SEFCE can give $\mathsf{P}_1$ more payoff than SSE. Also, every SSE is an SEFCE, but no vice versa.

See \cite{bovsansky2017computation} for a breakdown of computational complexity for different classes of games, (e.g., correlation/signaling (correlated, pure, behavioral), whether there is chance, and different levels of imperfect information).

\section{A Useful Way of Reasoning about SEFCE}
For readers who are more familiar with SSE or are uncomfortable with the `correlation' present in SEFCE, we give an easy interpretation of SEFCE in perfect information games. Given $\game$, consider a modified game $\game'$, where before a follower vertex $s$, we add a leader vertex $s'$ just before it with \textit{two} actions, where each action leads to a copy of the game rooted at the follower vertex $s$. This construction is repeatedly performed in a bottom-up fashion. After this entire process is completed, we will find the \textit{SSE} of $\game'$ (which is a much larger game than $\game$).

The only purpose of this leader vertex is to allow mixing between follower strategies. Now, the recommendation to the follower is explicit via $s'$. Each action in $s'$ corresponds to a single recommended action for the follower at $s$. Note that only a binary signal is needed (since mixing will only occur between at most 2 follower actions). Let $s_a$ and $s_b$ be duplicate follower vertices. Crucially, the leader is allowed to commit to different strategies for each subgame following $s_a$ and $s_b$. The SSE in $\game'$ can be mapped to the SEFCE in $\game$. Since in SSE, best responses are pure, and the probabilities leading to $s_a$ and $s_b$ (from the leader signaling node) give the probabilities at state $s$ in $\game$ for the (pure) action to be taken at $s_a$ and $s_b$.

We reiterate that this construction is not one that is practical, but rather one to help to gain intuition for the SEFCE. The solution using the constructed game is the same as the SEFCE by running the algorithm in Section~\ref{sec:review} and seeing how the $\bigwedge$ operator (for SSE) at the added root in $\game'$ mimics the follower vertex in SEFCE.

\begin{figure}
    \centering
    \begin{tikzpicture}[scale=1,font=\footnotesize]
\tikzstyle{leader}=[regular polygon,regular polygon sides=3,draw,inner sep=1.2,fill=black];
\tikzstyle{follower}=[regular polygon,regular polygon sides=3, rotate=180,draw,inner sep=1.2, fill=black];
\tikzstyle{terminal}=[draw,inner sep=1.2];

\tikzstyle{level 1}=[level distance=5mm,sibling distance=50mm]
\tikzstyle{level 2}=[level distance=5mm,sibling distance=25mm]
\tikzstyle{level 3}=[level distance=5mm,sibling distance=12mm]

\node(root)[leader,label=above: $\text{Augmented signaling vertex}$]{}
child{node(0)[follower,label=below:{$s_a$}]{}
    child{node(1)[leader, label=above left:$s_a'$]{}
        child{node(3)[terminal, label=below:{$(10,-1)$}]{}
        }
        child{node(4)[terminal, label=below:{$(-1,1)$}]{}
        }
        edge from parent node[above left]{stay}
    }
    child{node(2)[leader, label=above right:$s_a''$]{}
        child{node(5)[terminal, label=below:{$(0,0)$}]{}
        }
        child{node(6)[terminal, label=below:{$(4,0.5)$}]{}
        }
        edge from parent node[above right]{exit}
    }
}
child{node(0_)[follower,label=below:{$s_b$}]{}
    child{node(1_)[leader, label=above left:$s_b'$]{}
        child{node(3_)[terminal, label=below:{$(10,-1)$}]{}
        }
        child{node(4_)[terminal, label=below:{$(-1,1)$}]{}
        }
        edge from parent node[above left]{stay}
    }
    child{node(2_)[leader, label=above right:$s_b''$]{}
        child{node(5_)[terminal, label=below:{$(0,0)$}]{}
        }
        child{node(6_)[terminal, label=below:{$(4,0.5)$}]{}
        }
        edge from parent node[above right]{exit}
    }
};
\end{tikzpicture}
    \caption{Example of duplicated follower vertices based off the game in Figure~\ref{fig:sefce_eg}.}
    \label{fig:dup-EFCE}
\end{figure}
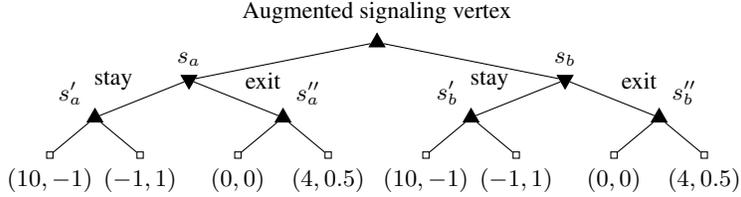

\section{Implementation Details}
\subsection{Borrowing Techniques from RL and FVI}
We employ target networks \cite{mnih2015human}. Instead of performing gradient descent on the `true' loss, we create a `frozen' copy of the network which we use the compute $\tilde{U}'_s$ (i.e., $\tilde{U}_s^{\text{target}}$). However, $\tilde{U}$ is still computed from the main network (with weights to be updated in gradient descent). The key idea is that the target $\tilde{U}'_s$ is no longer update at every epoch, which can destabilize training. We update the target network with the main network once every $2000$ episodes.

For larger games, we noticed that the bulk of loss was attributed to a small fraction of states. To focus attention on these states, we employ prioritized replay \cite{schaul2015prioritized}. We set the probability of selecting each state $s$ in the replay buffer to be proportionate to the square root of the last loss, i.e., $L(U_s, U'_s)^\alpha$ observed. We used $\alpha=0.5$ for convenience (\cite{schaul2015prioritized} suggest a value of 0.7).

Finding out the best hyperparameters for these add-ons is beyond the scope of this paper and left as future work.

\subsection{Modified Loss Function}
Our experience is that $L_\infty$ does manage to learn EPFs well, however, learning can be slow and sometimes unstable. Our hypothesis is that slow learning is due to the fact only the point responsible for the loss, as well as its neighbors has its coordinates updated during training. This very `local' learning of $\tilde{U}_s$ makes learning slow, particularly at the start of training. Second, we found that rather than $L_\infty$, using the \textit{square} of the largest absolute pointwise difference tends to stabilize training (though Theorem~\ref{thm:FA_guarantee} would have to be modified to be in terms of $\sqrt{\epsilon}$ instead). 

Let $U_s$ and $U'_s$ be two EPFs represented by $k_1$ and $k_2$ knots. Let $X_1 = \{ x_1, \dots x_{k_1} \}$ and $X_2 = \{ x'_1, \dots x'_{k_2} \}$ be the x-coordinates of the knots in $U_s$ and $U'_s$ respectively. Then, we use the following loss
\begin{align}
    \begin{split}
        L(U_s, U'_s) &= \sum_{x \in X_1 \cup X_2} \left( U_s(x) -U'_s(x) \right)^2.
    \end{split}
\end{align}
This loss still avoids costly promises since we are still taking pointwise differences (rather than over an integral). 

\subsection{Practical Implementation of Upper Concave Envelope and Left-Truncation}
Theoretically, finding the upper concave envelope of $k$ points can be found in linear time. However, this algorithm requires a significant number of backtracking and if-else statements, making this implementation unsuitable for batch operations on a GPU. The alternative which we employ runs in $\mathcal{O}(k^3)$ time which is in practice much faster when run on a GPU. For every distinct pairs of points $(x_i, y_i)$, and $(x_j, y_j)$, we check, for every point $(x_a, y_a)$ where $x_a \in [x_i, x_j]$ whether $(x_a, y_a)$ lies below or above the line segment $(x_i, y_i), (x_j, y_j)$. If a point $(x_a, y_a)$ is below \textit{any} such line segment, we flag it as `not included', indicating that the upper concave envelope will not include this point. The overall scheme is complicated (see attached code) but runs significantly faster than the linear time method when batch sizes are greater than 32. An important downside, however is (a) the amount of \textit{GPU} memory used for intermediate calculations and (b) the poor scaling (cubic) in terms of number of knots (which is $\beta m$, where $\beta$ is the branching factor and $m$ the number of knots) per EPF.

\paragraph{Dealing with different number of actions at each vertex and truncated points.} Rather than removing points and padding them (to make each batch fit nicely in rectangular tensor), we maintain a `mask' matrix which indicates that such a point is inactive. These points will not be used in computation of upper concave envelopes (both as potential points and as part of a line segment). Furthermore, this scheme makes it convenient to truncate points (simply mask those truncated points out and perform interpolation to get the new point on at $\tau(s')$).

\subsection{Training Only Using Decreasing Portions of EPF}
Kearning the increasing portion of an EPF is not useful, since points there are Pareto dominated. Hence, when extracting $\tilde{\pi}$ we will never select those points. For example in Figure~\ref{fig:demo2-3} and \ref{fig:demo2-4}, if there were parts of the EPF in the yellow regions, they would not matter since the leader would select the maximum point with x-coordinate at $0$ instead.

If we instead consider a slight variation of the EPF $U_s : \mathbb{R} \mapsto \mathbb{R} \cup \{ -\infty \}$ that gives the maximum leader payoff given the follower gets a payoff of \textit{at least} $\mu_2$ (rather than exactly $\mu_2$). This slight change ensures that EPFs are never increasing, while keeping all of the properties we proved earlier on. Omitting the increasing portions saves us from wasting any of the $m$ knots on the increasing portions, and instead focus on the decreasing portion (where there is a real trade-off between payoffs between $\mathsf{P}_1$ and $\mathsf{P}_2$. One example of this is shown in Figure~\ref{fig:epoch-comparison-100k} and Figure~\ref{fig:epoch-comparison-2M}, where we showed the `true' EPF and the modified EPFs that we use for our experiments. 

We describe what happens concretely. Let $\tilde{U}'_s$ be computed based on \eqref{eq:sefce_recurse} with its representation given by the set of knots $\{(x_1, y_1), \dots ,(x_k, y_k)\}$, assumed to be sorted in ascending order of $x$-coordinates, and where $x_1 = \LBV(s)$. Let the $j = \argmax_{i} y_i$. Then, the set of points which we use for training is the modified set $\{(\LBV(s), y_j), (x_j, y_j), (x_{j+1}, y_{j+1}), \dots (x_k, y_k)\}$.

\subsection{Sampling of Training Trajectories}
One of the design decisions in FVI is how one should sample states, or trajectories. In the single-player setting, it is commonplace to use some form of $\epsilon$-greedy sampling. In our work, we use an even simple sampling scheme which takes actions uniformly at random. 

There are a few exceptions. For \textsc{RC} sampled \textit{states} uniformly at random. This was made possible because the game was small and we could enumerate all states. The implication is that each leaf is sampled much more frequently than from actual trajectories. Our experience is that since the game is small, getting samples from uniform trajectories should still work well. For \textsc{Tantrum}, the game has a depth of $50$. In many cases, states in the middle are not learning anything meaningful because their children EPFs have not been learned well. As such, we adopt a `layered' approach, where initially we only allow for states $s$ at most $d_{\text{max}}$ to be added to the replay buffer, $d_\text{max}$ is gradually increased as training goes on. This helps EPFs to be learned for states deeper in the tree first before their parents.
We find that for \textsc{Tantrum} (the non-featurized version with $n=25)$, this was essential to get stable learning of EPFs (recall that in our setting, \textsc{Tantrum} has a size of roughly $\sim 3^{25}$ and a depth of $50$. A uniform trajectory leaves some states to be sampled with probability $1/2^{50}$). We start off at $d_{\text{max}}=20$ and reduce $d_\text{max}$ by $1$ for every $50000$ epochs. For other games, uniform trajectories work well enough since the game is not too deep.

\section{Additional Details on Experimental Setup}
\subsection{Environment Details}
For all our experiments, the network is a multilayer fully connected network of width $128$, depth $8$, ReLU activations and number of knots $m=8$. We used the PyTorch library \cite{paszke2017automatic} and a GPU to accelerate training. No hyperparameter tuning was done. 
\subsection{\textsc{RC} Map Generation Details}
Maps were generated with each reward map being drawn independently from a log Gaussian process (with query points given by the $(x, y)$ coordinates on the grid). We use the square-exponential kernel, a length scale of 2.0 and a standard deviation of 0.1. This way of generating maps was to encourage spatial smoothness in rewards for more realism. Figure~\ref{fig:sample_reward_maps} give examples of maps generated using this procedure. From the figures, one cans see that good regions for $\mathsf{P}_1$ may not be good for $\mathsf{P}_2$ and vice versa.

\begin{figure}
    \centering
    \begin{subfigure}{0.4 \textwidth}
    \includegraphics[width=\textwidth]{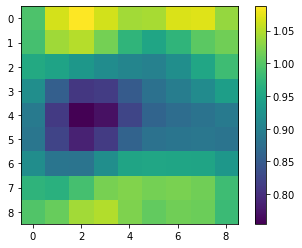}
    \phantomcaption{}
    \label{fig:sample_reward_map_rew_1}
    \end{subfigure}
    \begin{subfigure}{0.4 \textwidth}
    \includegraphics[width=\textwidth]{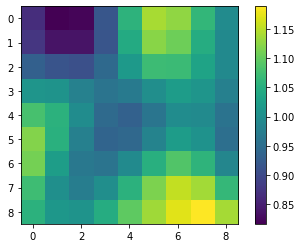}
    \phantomcaption{}
    \label{fig:sample_reward_map_rew_2}
    \end{subfigure}
    \caption{Left to right: An example of reward maps used for $\mathsf{P}_1$ and $\mathsf{P}_2$ in \textsc{RC}+.}
    \label{fig:sample_reward_maps}
\end{figure}

\subsection{\textsc{Tantrum} Generation}
For [\textsc{Tantrum}], the values of $q_2$ were chosen (somewhat arbitrarily) to be in $\{1.5, 2.1, 3.4, 5.1, 6.7\}$. 

\subsection{Training Hyperparameters}
We used the Adam optimizer (Kingma and Ba) with AMSGrad (Reddi et. al.) with a learning rate of 1e-5 (except for $\textsc{RC}$, where we found using a learning rate of 1e-4 was more suitable). We use the implementation provided in the PyTorch \cite{paszke2017automatic} library. The replay buffer was of a size of 1M. The minibatch size was set to 128. The target network's parameters was updated once every $2000$ training epochs. 
\subsection{Number of Epochs and Termination Criterion}
Unfortunately, it is very rare to have a small loss for every single state. Hence, we terminate training at a fixed iteration. The number of epochs are given in Table~\ref{tbl:expt-differences}.

\begin{table}[]
    \centering
    \begin{tabular}{ccccc}
         & Num Epochs & State sampling method & Prioritized replay & Traj. frequency \\
         \hline \hline
         \textsc{RC}& 2M & Random state & N/A & N/A\\
         \hline 
         \textsc{Tantrum}& 4M & Unif. trajectory, layered & Yes & 10\\
         \hline
         Feat. \textsc{Tantrum}& 2.7M & Uniform trajectory & Yes & 10\\
         \hline
         \textsc{RC}+ & 1.7M & Uniform trajectory & Yes & 20
    \end{tabular}
    \caption{Differences in experimental setups over each game. Traj. frequency refers to how many epochs before we sample a new trajectory.}
    \label{tbl:expt-differences}
\end{table}

\section{Qualitative Discussion of Optimal Strategies in \textsc{Tantrum}}
We explore \textsc{Tantrum} in the special case when $q_1=1$ and $q_2 > 1$. Intuitively, we should `use' as many threats as possible. That is achieved by the leader committing to $(-1, -1)$ for all future states. If $\textsf{P}_1$ does that from the beginning, it will give $\mathsf{P}_2$ $-n$ (the number of times the stage game is repeated) payoff to each player. Naturally, one upper-bound on how much $\textsf{P}_1$ can get is $n/q_2$, that is, $\textsf{P}_2$ chooses to accede on average of $n/q_2$ times per playthrough. $\mathsf{P}_1$ cannot possible get more since that would lead to to $\mathsf{P}_2$ losing more than $n$ (which is the worst possible threat the leader can make from the beginning). 

In our experiments, this was indeed true, and can be achieved by the following strategy. Let $\pi$ be such that (a) the $\mathsf{P}_1$ plays to $(0, 0)$ at all leader vertices and $\mathsf{P}_2$ accedes for the first $j = \lfloor n / q_2 \rfloor$ stages with probability $1$. At the $j+1$-th vertex, it plays a mixed strategy (or rather, it receives the recommendation to mix) strategies, with probability $(n - jq_2)/q_2$ it accedes. Clearly, the expected payoff for $\mathsf{P}_2$ is $-n$, and $\mathsf{P}_2$ cannot do any better.

However, this result does not hold for all settings, typically when $n$ is small. This is because we need to consider the threat is strong enough at every stage. Consider the case where $n=3$ and $q_2=2$. Our derivation suggests that at the first stage, the follower accedes with probability 1 and at the second stage, it accedes with probability $0.5$. 

At the first stage, $\mathsf{P}_2$ is indeed incentivized to accede. since if it will suffer from $-3$ if not, since acceding yields $-2$ payoff, which when combined with the expected payoff of $-1$ in the future, is equivalent to the threat of $-3$. At the second stage, it is just barely incentive compatible for the follower to accede. Specifically, if the player accedes, it will receive a payoff of $-4$ (it has already accumulated $-2$ from the previous stage). On the other hand, the grim trigger threat gives a payoff of only $-4$ ($-2$ from the past and $-2$ from the future). Hence, after receiving the recommendation $\mathsf{P}_2$ is just incentivized to not deviate. However, when $q_2$ is increased by just a little (say to $2.1$), this incentive is not sufficient. The \textbf{future} threat from not acceding is $-2$, but the follower already loses $2.1$ from acceding. \footnote{This phenomena is very similar to the difference between a coarse correlated equilibrium and a regular CE. The difference is that a player has to decide to deviate before or after receiving its recommended actions.}

In general, for our derived bound to be tight, we will require 
\begin{align*}
    \underbrace{n-\lfloor n/q \rfloor}_{\text{threat from grim trigger}} \geq \underbrace{q}_{\text{cost from acceding this time round}},
\end{align*}
that is, at the last accede recommendation (possibly with some probability), we still have enough rounds remaining as threats to maintain incentive compatibility. Technically we require this for all previous rounds; however this condition being satisfied for the last round implies that it is satisfied for all previous rounds. 

We can also see from this discussion that in this repeated setting, it is always beneficial to recommend accede to the follower higher up the tree; this way, there is more room for the leader to threaten the follower with future $(-1, -1)$ actions.

\paragraph{Featurized \textsc{Tantrum}}
As far as we know, there is no simple closed form solution for featurized \textsc{Tantrum}.
\twocolumn

\end{document}